%% file: main.tex
\documentclass[a4paper,autoref]{lipics-v2021}\nolinenumbers
\usepackage{amsmath,amssymb,amsfonts,latexsym}
\usepackage{color,graphicx}
\usepackage{url} 
\usepackage{fancybox}
\usepackage{algorithm}
\usepackage[noend]{algpseudocode}
\usepackage{mathtools}
\usepackage{enumerate,xspace}
\usepackage{siunitx} 
\hideLIPIcs

\counterwithin{lemma}{section}
\counterwithin{theorem}{section}
\counterwithin{proposition}{section}
\counterwithin{claim}{section}
\counterwithin{observation}{section}
\counterwithin{corollary}{section}

\newcommand{\eps}{\varepsilon}
\renewcommand{\epsilon}{\eps}
\newcommand{\etal}{\emph{et al.}\xspace}

\theoremstyle{plain}

\newenvironment{myquote}%
  {\list{}{\leftmargin=4mm\rightmargin=4mm}\item[]}%
  {\endlist}
\newenvironment{claiminproof}{\begin{myquote}\noindent\emph{Claim.}}{\end{myquote}}
\newenvironment{proofinproof}{\begin{myquote}\noindent\emph{Proof.}}{\hfill $\lhd$ \end{myquote}}

\newcommand{\G}{\ensuremath{\mathcal{G}}}

\renewcommand{\leq}{\leqslant}
\renewcommand{\geq}{\geqslant}

\DeclarePairedDelimiter\ceil{\lceil}{\rceil}
\DeclarePairedDelimiter\floor{\lfloor}{\rfloor}

\DeclareMathOperator{\dgr}{deg}
\newcommand{\opt}{\mbox{{\sc opt}}\xspace}

\newcommand{\alg}{\mbox{{\sc alg}}\xspace}

\newcommand{\Sold}{S_{\mathrm{old}}}
\newcommand{\Snew}{S_{\mathrm{new}}}

\newcommand{\Gexp}{G_{\mathrm{exp}}}
\newcommand{\optout}{\opt_{\mathrm{out}}}
\newcommand{\Nout}{N_{\mathrm{out}}}

\newcommand{\Salg}{S_{\mathrm{alg}}}
\newcommand{\Sopt}{S_{\mathrm{opt}}}
\newcommand{\Dalg}{D_{\mathrm{alg}}}
\newcommand{\Dopt}{D_{\mathrm{opt}}}
\newcommand{\Ialg}{I_{\mathrm{alg}}}

\newcommand{\maxopt}{\mbox{{\sc max-opt}}\xspace}
\newcommand{\tprev}{t_{\mathrm{prev}}}
\newcommand{\tnext}{t_{\mathrm{next}}}

\newcommand{\mis}{{\sc Independent Set}\xspace}
\newcommand{\domset}{{\sc Dominating Set}\xspace}
\newcommand{\fvs}{{\sc Feedback Vertex Set}\xspace}
\newcommand{\vc}{{\sc Vertex Cover}\xspace}

\title{Stable Approximation Algorithms for Dominating Set and Independent Set }
 \funding{MdB, AS, and FS are supported by the  Dutch Research Council (NWO) through Gravitation-grant NETWORKS-024.002.003.}
\author{Mark de Berg}{Department of Mathematics and Computer Science, TU Eindhoven, the Netherlands}{M.T.d.Berg@tue.nl}{}{}
\author{Arpan Sadhukhan}{Department of Mathematics and Computer Science, TU Eindhoven, the Netherlands}{A.Sadhukhan@tue.nl}{}{}
\author{Frits Spieksma}{Department of Mathematics and Computer Science, TU Eindhoven, the Netherlands}{f.c.r.spieksma@tue.nl}{}{}
\authorrunning{M.~de Berg, A.~Sadhukhan and F.~Spieksma} 
\Copyright{Mark de Berg and Arpan Sadhukhan and Frits Spieksma}



\begin{document}
\setcounter{page}{0}
\maketitle

\begin{abstract}
We study \domset and \mis for dynamic graphs in the vertex-arrival model.
We say that a dynamic algorithm for one of these problems is \emph{$k$-stable} 
when it makes at most $k$ changes to its output independent set or dominating set
upon the arrival of each vertex. We study trade-offs between the stability 
parameter~$k$ of the algorithm and the approximation ratio it achieves. 
We obtain the following results.  
\begin{itemize}
\item We show that there is a constant $\eps^*>0$ such that any dynamic $(1+\eps^*)$-approximation
      algorithm  for \domset has stability parameter~$\Omega(n)$,
      even for bipartite graphs of maximum degree~4.   
\item  We present algorithms with very small stability parameters for \domset in the 
       setting where the arrival degree of each vertex is upper bounded by $d$. In particular, 
       we give a $1$-stable $(d+1)^2$-approximation algorithm, a $3$-stable $(9d/2)$-approximation algorithm, and an $O(d)$-stable $O(1)$-approximation algorithm.
\item We show that there is a constant $\eps^*>0$ such that any dynamic $(1+\eps^*)$-approximation algorithm 
      for \mis has stability parameter~$\Omega(n)$,
      even for bipartite graphs of maximum degree~3. 
\item Finally, we present a 2-stable $O(d)$-approximation algorithm for \mis, in the setting where the
      average degree of the graph is upper bounded by some constant $d$ at all times. We extend this latter algorithm to the fully dynamic model where vertices can also be deleted, achieving a 6-stable $O(d)$-approximation algorithm.
\end{itemize}
\end{abstract}

\input{introduction.tex}
\input{SASvsPTAS}
\input{Dominating_Set}
\input{constant_approx_domset}
\input{Independent_Set}

\input{Deletions_Independent_Set}

\section{Concluding remarks}
\label{sec:conclusions}
We have studied the stability of dynamic algorithms for \domset and \mis in the vertex-arrival
model. For both problems we show that a SAS does not exist. For \mis this even holds when the 
degrees of all vertices are bounded by 3 at all times. This is clearly tight, since a SAS is easily 
obtained for graphs with maximum degree~2. For \domset the no-SAS result holds for degree-4 graphs.
A challenging open problem is whether a SAS exists for \domset for degree-3 graphs. 
We have also described algorithms whose approximation ratio and/or stability depends on the (arrival or average) degree.
An interesting open problem here is: Is there a 1-stable $O(d)$-approximation algorithm for \domset,
when the arrival degree is at most~$d$? Finally, we believe the concept of stability for dynamic algorithms,
which purely focuses on the change in the solution (rather than computation time) is
interesting to explore for other problems as well.

\bibliography{references}

\appendix

\end{document}

%% file: introduction.tex
\section{Introduction}
A {\em dominating set} in a simple, undirected graph $G=(V,E)$ is a subset $D \subset V$ 
such that each vertex in $V$ is either a neighbor of a vertex in $D$, or it is in $D$ itself. 
An {\em independent set} is a set of vertices $I \subset V$ such that no two vertices in $I$
are neighbors. \domset (the problem of finding a minimum-size dominating set) and \mis 
(the problem of finding a maximum-size independent set) are fundamental problems in algorithmic 
graph theory. They have numerous applications and served as prototypical problems 
for many algorithmic paradigms.

We are interested in \domset and \mis in a dynamic setting, where the graph~$G$ changes over time.
In particular, we consider the well-known \emph{vertex-arrival model}. Here one starts with an empty
graph~$G(0)$, and new vertices arrive one by one, along with their edges to previously arrived vertices. 
In this way, we obtain a sequence of graphs $G(t)$, for $t=0,1,2,\ldots$. Our algorithm is then required
to maintain a valid solution---a dominating set, or an independent set---at all times.
In the setting we have in mind, computing a new solution is not the bottleneck,
but each change to the solution (adding or deleting a vertex from the solution) is expensive.
Of course we also want that the maintained solution has a good approximation ratio.
To formalize this, and following De~Berg~\etal~\cite{DBLP:conf/swat/BergSS22}, we say that a dynamic algorithm is a 
\emph{$k$-stable $\rho$-approximation algorithm} if, upon the arrival of each vertex,
the number of changes (vertex additions or removals) to the solution is at most~$k$
and the solution is a $\rho$-approximation at all times. In this framework, we study
trade-offs between the stability parameter~$k$ and the approximation ratio that can be achieved. 
Ideally, we would like to have a so-called \emph{stable approximation scheme (SAS)}:
an algorithm that, for any given yet fixed parameter $\epsilon > 0$ is $k_{\epsilon}$-stable 
and gives a $(1+\epsilon)$ approximation algorithm, where $k_{\epsilon}$ 
only depends on $\epsilon$ and not on the size of the current instance.
(There is an intimate relationship between local-search PTASs and SASs;
we come back to this issue in Section~\ref{sec:SASvsPTAS}.)

The vertex-arrival model is a standard model for online graph algorithms, and our
stability framework is closely related to online algorithms with bounded recourse.
However, there are two important differences. First, computation time
is free in our framework---for instance, the algorithm may decide to compute an optimal solution 
in exponential time upon each insertion---while in online algorithms with bounded 
recourse the update time is typically taken into account. Thus we can fully focus on the 
trade-off between stability and approximation ratio. Secondly, we consider
the approximation ratio of the solution, while in online algorithm one typically
considers the competitive ratio. Thus we compare the quality of our solution
at time $t$ to the \emph{static optimum}, which is simply the optimum for the graph~$G(t)$. 
A competitive analysis, one the other hand, compares the quality of the
solution at time~$t$ to the \emph{offline optimum}: the best
solution for $G(t)$ that can be computed by a dynamic  algorithm
that knows the sequence $G(0),\ldots,G(t)$ in advance but must still 
process the sequence with bounded recourse. See also the discussion in
the paper by Boyar~\etal~\cite{DBLP:journals/algorithmica/BoyarEFKL19}.
Thus approximation ratio is a much stronger notion that competitive ratio.
As a case in point, consider \domset, and suppose that $n$ singleton vertices
arrive, followed by a single vertex with edges to all previous ones. Then
the static optimum for the final graph is~1, while the offline optimum
with bounded recourse is~$\Omega(n)$.

\subparagraph*{Related work.}
We now review some of the most relevant existing literature on the online version of our problems.
(Borodin and El-Yaniv~\cite{DBLP:books/daglib/0097013} give a general introduction to online computation.) 
The classical online model, where a vertex that has been added to the dominating set or to the independent set, 
can never be removed from it---that is, the no-recourse setting---is e.g. considered by King and Tzeng~\cite{DBLP:journals/ipl/KingT97} 
and Lipton and Tomkins~\cite{DBLP:conf/soda/LiptonT94}. They show that already for the special case of interval graphs 
no online algorithm has constant competitive ratio; see also De~\etal~\cite{DBLP:journals/corr/abs-2111-07812},
who study these two problems for geometric intersection graphs. For \domset in the vertex-arrival model, 
Boyar \etal~\cite{DBLP:journals/algorithmica/BoyarEFKL19} give online algorithms with bounded competitive ratio for 
trees, bipartite graphs, bounded degree graphs, and planar graphs. They actually analyze both the competitive ratio 
and the approximation ratio (which, as discussed earlier, can be quite different).
A crucial difference between the work of Boyar \etal~\cite{DBLP:journals/algorithmica/BoyarEFKL19} and ours 
is that they do not allow recourse: in their setting, once a vertex is added to the dominating set, it cannot be removed.
 
To better understand online algorithmic behavior, various ways to relax classical online models have been studied. 
In particular, for \mis, among others, Halld\'orsson \etal~\cite{DBLP:journals/tcs/HalldorssonIMT02} 
consider the option for the online algorithm to maintain multiple solutions. 
G\"obel \etal~\cite{DBLP:conf/icalp/0002HKSV14} analyze a stochastic setting where, among other variants, 
a randomly generated graph is presented in adversarial order (the prophet inequality model), 
and one where an adversarial graph is presented in random order. 
In these settings they find constant competitive algorithms for \mis on interval graphs, 
and more generally for graphs with bounded so-called inductive independence number. 

One other key relaxation of the online model is to allow recourse. Having recourse can be seen as relaxing the irrevocability assumption in classical online problems, and allows to assess the impact of this assumption on the competitive ratio, see Boyar \etal~\cite{DBLP:journals/algorithmica/BoyarFKL22}. The notion of (bounded) recourse has been investigated for a large set of problems. Without aiming for completeness, we mention Angelopoulos \etal~\cite{DBLP:conf/mfcs/0001DJ18} and Gupta \etal~\cite{DBLP:conf/soda/GuptaKS14} who deal with online matching and matching assignments, Gupta \etal~\cite{DBLP:conf/stoc/GuptaK0P17} who investigate the set cover problem, and Berndt \etal~\cite{DBLP:conf/esa/BerndtEJLMR19} who deal with online bin covering with bounded migration, and Berndt \etal~\cite{berndt2019robust} who propose a general framework for dynamic packing problems. For these problems, it is shown what competitive or approximation ratios can be achieved when allowing a certain amount of recourse (or migration). Notice that, in many cases, an amortized interpretation of recourse is used; then the average number of changes to a solution is bounded (instead of the maximum, as for $k$-stable algorithms defined above). For instance, Lsiu and Charington-Toole~\cite{10.1007/978-3-031-18367-6_7} show that, for \mis, there is an interesting trade-off between the competitive ratio and the amortized recourse cost: for any $t > 1$, they provide a $t$-competitive algorithm for \mis using $t/(t-1)$ recourse cost. Their results however do not apply to our notion of stability. 

\subparagraph*{Our contribution.}
We obtain the following results. 
\begin{itemize}
\item In Section 2, we show that the existence of a local-search PTAS for the static 
      version of a graph problem implies, under certain conditions, the existence of a 
      SAS for the problem in the vertex-arrival model (whereas the converse need not be true). This implies that for graphs with strongly sublinear separators, a SAS exists for \mis and for \domset when the arrival degrees---that is, the degrees of the vertices upon arrival---are bounded by a constant.
\item In Section 3, we consider \domset in the vertex-arrival model. Let $d$ denote the
      maximum arrival degree. We show (i) there does not exist a SAS even for bipartite graphs of maximum degree~4, (ii) there is a 1-stable $(d+1)^2$-approximation algorithm, (iii) there is a 3-stable $(9d/2)$-approximation algorithm, and (iv) there is an $O(d)$-stable $O(1)$-approximation algorithm.
\item In Section 4, we consider \mis in the vertex-arrival model. We show that there does not exist a SAS even for bipartite graphs of maximum degree 3. Further, we give a 2-stable $O(d)$-approximation algorithm for the case where the average degree of $G(t)$ is bounded by~$d$ at all times. We extend this algorithm to the fully dynamic model (where vertices can also be deleted from the graph), finding a 6-stable $O(d)$-approximation algorithm.
\end{itemize}

%% file: SASvsPTAS.tex
\section{Stable Approximation Schemes versus PTAS by local search}
\label{sec:SASvsPTAS}

In this section we discuss the relation between Stable Approximation Schemes (SASs)
and Polynomial-Time Approximation Schemes (PTASs). Using known results on
local-search PTASs, we then obtain SASs for \mis and \domset on certain graph classes.
While the results in this section are simple, they set the stage for our main results
in the next sections.
\medskip

The goals of a SAS and a PTAS are the same: both aim to achieve an approximation ratio
$(1+\eps)$, for any given $\eps>0$. A SAS, however, works in a dynamic setting
with the requirement that $k_\eps$, the number of changes per update, is a constant 
for fixed~$\eps$, while a PTAS works in a static setting with the
condition that the running time is polynomial for fixed~$\eps$. Hence, there 
are problems that admit a PTAS (or are even polynomial-time solvable) but no SAS~\cite{DBLP:conf/swat/BergSS22}.

One may think that the converse should be true: if a dynamic problem admits a SAS, then 
its static version admits a PTAS. Indeed, we can insert the input elements one by one 
and let the SAS maintain a $(1+\eps)$-approximation, by performing at most $k_{\eps}$ 
changes per update. For a SAS, there is no restriction on the time needed to 
update the solution, but it seems we can simply try all possible combinations of 
at most $k_{\eps}$~changes in, say, $n^{O(k_{\eps})}$ time. This need not work, however, 
since there could be many ways to update the solution using at most $k_{\eps}$ changes. 
Even though we can find all possible combinations in polynomial time, we may not be able to decide 
which combination is the right one: the update giving the best solution
at this moment may get us stuck in the long run.
A SAS can avoid this by spending exponential time to decide what the right update is.
Thus the fact that a problem admits a SAS does not imply that it admits a PTAS.
\medskip

Notwithstanding the above, there is a close connection between SASs and PTASs and, 
in particular between SASs and so-called local-search PTASs: under certain conditions, 
the existence of a local-search PTAS implies the existence of a SAS. 
For simplicity we will describe this for graph problems, but the technique may be 
applied to other problems.

Let $G=(V,E)$ be a graph and suppose we wish to select a minimum-size (or maximum-size)
subset $S\subset V$ satisfying a certain property. Problems of this type include
\domset, \mis, \vc, \fvs, and more. 
A local-search PTAS\footnote{See the paper by
Antunes~\etal~\cite{DBLP:conf/esa/AntunesMM17} for a nice exposition on problems solved using local-search PTAS.}  
for such a graph minimization problem starts with an arbitrary feasible
solution~$S$---the whole vertex set~$V$, say---and then it tries to repeatedly decrease the size
of~$S$ by replacing a subset $\Sold\subset S$ by a subset $\Snew\subset V\setminus S$ 
such that $|\Snew|=|\Sold|-1$ and $(S\setminus \Sold)\cup \Snew$ is still feasible. 
(For a maximization problem we require $|\Snew|=|\Sold|+1$.)
This continues until no such replacement can be found.
A key step in the analysis of a local-search PTAS 
is to show the following, where $n$ is the number of vertices.
\begin{itemize}
    \item \textbf{Local-Search Property.} If $S$ is a feasible solution that is not a $(1+\eps)$-approximation
    then there are subsets $\Sold,\Snew$ as above with $|\Sold|\leq f_\eps$, 
    for some $f_\eps$ depending only on $\eps$ and not on~$n$. 
\end{itemize}
This condition indeed gives a PTAS, because we can simply try all possible pairs $\Sold,\Snew$, 
of which there are $O(n^{2 f_\eps})$. 

Now consider a problem that has the Local-Search Property in the vertex-arrival model, possibly
with some extra constraint (for example, on the arrival degrees of the vertices). 
Let $G(t)$ denote the graph after the arrival of the $t$-th vertex, and let $\opt(t)$ denote
the size of an optimal solution for~$G(t)$.
We can obtain a SAS if the problem under consideration has
the following properties.
\begin{itemize}
\item \textbf{Continuity Property.}
      We say that the dynamic problem (in the vertex-arrival model, possibly with extra constraints) is \emph{$d$-continuous} if $|\opt(t+1)-\opt(t)|\leq d$.
    In other words, the size of an optimal solution does not change by more
    than $d$ when a new vertex arrives. Note that the solution itself may change
    completely; we only require its size not to change by more than~$d$.
\item \textbf{Feasibility Property.} For maximization problems we require that any feasible
    solution for $G(t-1)$ is also a feasible solution for $G(t)$, and for minimization
    problems we require that any feasible solution for $G(t-1)$ 
     can be turned into a feasible solution for $G(t)$ by adding the arriving vertex
     to the solution. (This condition can be relaxed to saying that
     we can repair feasibility by $O(1)$ modifications to the current solution,
     but for concreteness we stick to the simpler formulation above.)
\end{itemize}
Note that \mis, \vc, and \fvs  are 1-continuous, and that \domset is $(d-1)$-continuous when
the arrival degree of the vertices is bounded by~$d\geq 2$. 
Moreover, these problems all have the Feasibility Property.

For problems that have the Local-Search Property as well as the Continuity and Feasibility
Properties, it is easy to give a SAS. 
Hence, non-existence of SAS for a problem with Continuity and Feasibility directly implies
non-existence of local-search PTAS.\footnote{\domset and \mis,
even with maximum degree bounded by 3, do not admit a PTAS assuming 
$\textsc{p}\neq \textsc{np}$~\cite{DBLP:journals/iandc/ChlebikC08,DBLP:conf/stoc/Zuckerman06}.
In sections~\ref{subsec:no-SAS-domset} and~\ref{subsec:no-SAS-mis}, by proving the non-existence of 
a SAS for \domset with maximum degree bounded by 4 
and \mis with maximum degree bounded by 3, we thus show non-existence of local-search PTAS,
independent of the assumption of $\textsc{p}\neq \textsc{np}$.}
 We give the pseudocode for minimization problems, but for
maximization problems a similar approach works.
\begin{algorithm}[H] 
\caption{\textsc{SAS-for-Continuous-Problems}($v$)}
\begin{algorithmic}[1]
\State $\rhd$ $v$ is the vertex arriving at time~$t$
\State $\Salg \gets \Salg(t-1) \cup \{v\}$ \hfill $\rhd$ $\Salg$ is feasible for $G(t)$ by the feasibility condition \label{step:feasible}
\While {$\Salg$ is not a $(1+\eps)$-approximation} 
       \State \parbox[t]{125mm}{Find sets $\Sold\subset \Salg$ and $\Snew\subset V(t)\setminus \Salg$
              with $|\Sold|\leq f_\eps$ and $|\Snew|=|\Sold|-1$ such that $(\Salg\setminus \Sold)\cup\Snew$
              is feasible, and set $\Salg \gets (\Salg\setminus \Sold)\cup\Snew$.}
 \EndWhile
\State $\Salg(t) \gets \Salg$
\end{algorithmic}
\end{algorithm}

\begin{theorem}\label{thm:local-PTAS-to-SAS}
Any $d$-continuous graph problem that has the Feasibility Property
and the Local-Search Property admits a SAS in the vertex arrival model,
with stability parameter  
$\left( \left\lceil (1+\eps)d \right\rceil+1 \right)\cdot (2f_\eps-1)$ for minimization problems and $d \cdot (2f_\eps+1)$
for maximization problems.
\end{theorem}
\begin{proof}
First consider a minimization problem.
By the Feasibility Property and the working of the algorithm, the solution that
\textsc{SAS-for-Continuous-Problems} computes upon the arrival of a new vertex is feasible. 
Moreover, it must end with a $(1+\eps)$-approximation because of
the while-loop; note that the sets $\Sold$ and $\Snew$ exist by the Local-Search Property.
Before the while-loop we add $v$ to $\Salg$, so just before the while-loop we have $|\Salg| = |\Salg(t-1)|+1$.
In each iteration of the while-loop, at most 
$f_\eps$ vertices are deleted from $\Salg$ and at most $f_\eps - 1$ vertices are added,
so the total number of changes is $(2f_\eps - 1) \cdot \mbox{(number of iterations)}$.
We claim that the number of iterations is at most $\left\lceil (1+\eps)d \right\rceil+1$. Indeed, 
before the arrival of $v$
we have $|\Salg(t-1)| \leq (1+\eps) \cdot \opt(t-1)$, and we have 
$|\opt(t-1)-\opt(t)|\leq d$ because the problem is $d$-continuous.
Hence, after $\left\lceil (1+\eps)d \right\rceil+1$ iterations we have
\[
\begin{array}{lll}
|\Salg| & = & |\Salg(t-1)|+1-\left( \left\lceil (1+\eps)d \right\rceil+1 \right) \\
      & \leq & (1+\eps) \cdot \opt(t-1) - (1+\eps)d \\
      & \leq & (1+\eps) \cdot \left( \opt(t)+ d\right)  - (1+\eps)d \\
      & \leq & (1+\eps) \cdot \opt(t).
\end{array}
\]
For a maximization problem the proof is similar. The differences are that 
we do not add an extra vertex to $\Salg$ in step~\ref{step:feasible},
that the number of changes per iteration is at most $2 f_\eps+1$,
and that the number of iterations is at most~$d$. Indeed, after $d$ iterations we have
\[
\begin{array}{lll}
|\Salg| & = & |\Salg(t-1)|+d \\
      & \geq & (1-\eps) \cdot \opt(t-1) +d \\
      & \geq & (1-\eps) \cdot \left( \opt(t) - d\right)  + d \\
      & \geq & (1-\eps) \cdot \opt(t).
\end{array}
\]
\end{proof}
This general result allows us to obtain a SAS for a variety of problems, for graph
classes for which a local-search PTAS is known.

Recall that a balanced \emph{separator} of a graph $G=(V,E)$ with $n$ vertices
is a subset $S\subset V$ such that $V\setminus S$ can be partitioned into subsets
$A$ and $B$ with $|A|\leq 2n/3$ and $|B|\leq 2n/3$ and no edges between~$A$ and $B$. 
We say that a graph class\footnote{We only consider hereditary graph classes, that is,
graph classes~$\G$ such that any induced subgraph of a graph in $\G$ is also in~$\G$.}~$\G$ 
has \emph{strongly sublinear separators}, if any graph~$G\in\G$
has a balanced separator of size $O(n^{\delta})$, for some fixed constant~$\delta<1$.
Planar graphs, for instance, have separators of size~$O(\sqrt{n})$~\cite{LT-planar-separator-thm}.
A recent generalization of separators are so-called \emph{clique-based separators},
which are separators that consist of cliques and whose size is measured in terms
of the number of cliques~\cite{bbkmz-ethf-20}. (Actually, the cost of a separator $S$ 
that is the union of cliques $C_1,\ldots,C_k$ is defined as $\sum_{i=1}^k \log (|C_i|+1)$,
but this refined measure is not needed here.) 
Disk graphs, which do not have a normal separator of sublinear size,
have a clique-based separator of size $O(\sqrt{n})$, for instance, 
and pseudo-disk graphs have a clique-based separator of size $O(n^{2/3})$~\cite{DBLP:conf/isaac/BergKMT21}.
For graph classes with strongly sublinear separators there are local-search
PTASs for several problems. Combining that with the technique above gives the following
result.
\begin{corollary} The following problems admit a SAS in the vertex-arrival model.
\begin{enumerate}[(i)]
\item \mis on graph classes with sublinear clique-based separators. 
\item \domset on graph class with sublinear separators, when the arrival degree of each vertex bounded by
       some fixed constant~$d$. 
\end{enumerate}
\end{corollary}
\begin{proof}
As noted earlier, \mis is 1-continuous and \domset is $(d-1)$-continuous. Moreover, these
problems have the Feasibility Property. It remains to check the Local-Search Property.
\begin{enumerate}[(i)]
\item
Any graph class with a separator of size $O(n^{\delta})$
has the Local-Search Property for \mis; see the paper by Her-Peled and Quanrud~\cite{DBLP:conf/esa/Har-PeledQ15}.
(In that paper they show the Local-Search Property for graphs of polynomial expansion---see
Corollary~26 and Theorem~3.4---and graphs of polynomial expansion have sublinear separators.) 
Theorem~\ref{thm:local-PTAS-to-SAS} thus implies the result
for such graph classes. To extend this to clique-based separators, we note
that (for \mis)  we only need the Local-Search Property for graphs that are the
union of two independent sets, namely the independent set~$S$ and an optimal independent set~$\Sopt$.
Such graphs are bipartite, so the largest clique has size two. Hence, the existence of
a clique-based separator of size $O(n^{\delta})$ immediately implies the existence of
a normal separator of size $O(n^{\delta})$.
\item For \domset on graphs with polynomial expansion (hence, on graphs with sublinear separators) the Local-Search Property holds~\cite[Theorem 3.15]{DBLP:conf/esa/Har-PeledQ15}. Theorem~\ref{thm:local-PTAS-to-SAS} thus implies an
$O(d \cdot f_\eps)$-stable $(1+\eps)$-approximation algorithm, for some constant $f_\eps$
depending only on $\eps$. Hence, if $d$ is a fixed constant , we obtain a SAS.
\end{enumerate}
\end{proof}

%% file: Dominating_Set.tex
\section{\domset}
In this section we study stable approximation algorithms for \domset in the vertex
arrival model. We first show that the problem does not admit a SAS, even when the 
maximum degree of the graph is bounded by~4. After that we describe three
algorithms that achieve constant approximation ratio with constant stability,
each with a different trade-off between approximation ratio and stability, 
in the setting where each vertex arrives with constant degree.

Let $G=(V,E)$ be a graph. We denote the open neighborhood of a 
subset~$W\subset V$ in $G$ by $N_G(W)$, so
$N_G(W) := \{ v\in V \setminus W: \mbox{there is a $w\in W$ with $(v,w)\in E$} \}$.
The closed neighborhood $N_G[W]$ is defined as $N_G(W)\cup W$. When the graph~$G$ is clear from
the context, we may omit the subscript~$G$ and simply write $N(W)$ and $N[W]$.

\subsection{No SAS for graphs of maximum degree 4}\label{subsec:no-SAS-domset}
Our lower-bound construction showing that \domset does not admit a SAS---in fact, our
construction will show a much stronger result, namely that there is a fixed constant $\eps^*>0$
such that any stable $(1+\eps^*)$-approximation algorithm must have stability
$\Omega(n)$---is based on the existence of a certain type of expander graphs, as given by the following proposition.
Note that $L$ (the left part of the bipartition of
the vertex set) is larger by a constant fraction than $R$ (the right part of the bipartition), 
while the expansion property goes from $L$ to~$R$.
The proof of the proposition was communicated to us by Noga Alon~\cite{alon2023}.
\begin{proposition} \label{prop:alon-expander}
For any $\mu > 0$ and any $n$ that is sufficiently large,
there are constants $\eps,\delta$ with $0<\eps<1$ and $0<\delta<1$ such that there is a 
bipartite graph $\Gexp(L\cup R,E)$ with the following properties:
\begin{itemize}
    \item  $|L| = (1 + \epsilon)n$ and $|R| = n$.
    \item The degree of every vertex in $G$ is at most $3$.
    \item For any $S\subset L$ with $|S| \leq \delta n$ we have $|N(S)| \geq (2-2\mu)|S|$.
\end{itemize}
\end{proposition}
\begin{proof}
Let $\mu> 0$ and let $t$ be an integer so that $t > 1/ \mu$. 
Let  $\epsilon \leq 1/3^{2t+1}$ be a fixed positive number. 
Let $H=(A\cup B, E_H)$ be a $3$-regular bipartite graph with vertex classes $A$ and $B$, 
each of size $(1+\epsilon)n$, in which for every subset $S\subset B$ of size at
most $\delta n$ we have $|N(S)|\geq (2-\mu)|S|$, for some fixed real number $\delta = \delta(\mu) > 0$.
(It is known that random cubic bipartite graphs have this property with high probability.)
Now pick a set $T\subset A$ of $\epsilon n$ vertices so that the distance between 
any two of them is larger that $2t$. Such a set exists, as one can choose its members one by one, 
making sure to avoid the balls of radius~$2t$ around the already chosen vertices.
This gives a set $T$ of the desired size since $\epsilon 3^{2t+1} < 1$. 
We define $\Gexp$ to be the induced subgraph of $H$ on the classes of vertices $R = A\setminus T$ and $L = B$.
It remains to show that $\Gexp$ has the desired properties.

We have $|L| = (1 + \epsilon)n$ and $|R| = n$ by construction, and the maximum degree of $\Gexp$ is 
clearly at most $3$. Now let $S$ be a set of at most $\delta n$ vertices in $L$. We
have to show that $N_{\Gexp}(S)$, its neighbor set in $\Gexp$, has size at least $(2-2\mu)|S|$. 
We can assume without loss of generality that $S \cup N_{\Gexp}(S)$ is connected in $\Gexp$, 
since we can apply the bound to each connected component separately and just add
the inequalities. Note that this assumption implies that $S \cup N_H(S)$ is also connected in $H$. 
Observe that the neighborhood $N_H(S)$ of $S$ in the original graph $H$, which is contained in $B$, 
is of size at least $(2-\mu)|S|$, by the property of $H$. If the set~$T$ of deleted vertices has at most $\mu|S|$
members in $N_H(S)$ then the desired inequality holds and we are done. Otherwise $T$ has more than $\mu|S|$ vertices that belong to $N_H(S)$. 
But the distance between any two such vertices in $H$ is larger than $2t$ 
(and so the balls of radius $t$ around them are disjoint), and since $S \cup N_H(S)$ 
is connected this would imply  $|S|> t|T \cap N_H(S)| > t \mu|S| > |S|$,
which is a contradiction.
\end{proof}
Now consider \domset in the vertex-arrival model. Let $G(t)$ denote
the graph at time $t$, that is, after the first $t$ insertions. 
Let $\eps^*>0$ be such that $\eps^*< \min \left(\frac{\eps}{2+\eps}, \frac{0.49\delta}{2(1+\eps)} \right)$,
where $\eps$ and $\delta$ are the constants in the expander construction of Proposition~\ref{prop:alon-expander}.
Consider a dynamic algorithm \alg for \domset such that $|\Dalg(t)| \leq (1+\eps^*)\cdot \opt(t)$
at any time~$t$, where $\Dalg(t)$ is the output dominating set of \alg at time~$t$
and $\opt(t)$ is the minimum size of a dominating set for~$G(t)$. Let $f_{\eps^*}(n)$ denote 
the stability of \alg, that is, the maximum number of changes it performs on $\Dalg$ when 
a new vertex arrives, where $n$ is the number of vertices before the arrival. 

We now give a construction showing that, for arbitrarily large $n$, there is a
sequence of $n$ arrivals that requires $f_{\eps^*}(n) \geq \frac{1}{6(7+6\eps)}\lfloor{\delta n}\rfloor$.
To this end, choose $N$ large enough such that the bipartite expander graph $\Gexp=(L\cup R,E)$ from Proposition~\ref{prop:alon-expander} 
exists for $\mu=0.005$ and $|R|=N$. (The value for $\mu$ has no special significance
and is simply a conveniently small constant.) Label the vertices in $L$ as
$\ell_1,\ldots, \ell_{|L|}$ and the vertices in $R$ as $r_1,\ldots, r_{|R|}$.

Our construction uses five layers of vertices, arriving one by one, as described
next and illustrated in Fig.~\ref{fig:domset-lowerbound}.
\begin{figure}
\begin{center}
\includegraphics{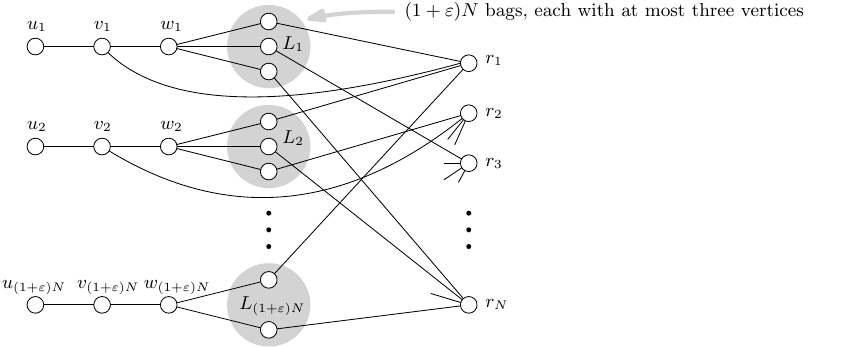}
\end{center}
\caption{The lower-bound construction for \domset.}
\label{fig:domset-lowerbound}
\end{figure}
\begin{itemize}
    \item \emph{Layer 1:} The first layer consists of vertices $u_1,\ldots,u_{(1+\eps)N}$,
                          each arriving as a singleton.
    \item \emph{Layer 2:} The second layer consists of vertices $v_1,\ldots,v_{(1+\eps)N}$,
                          where each $v_i$ has an edge to vertex~$u_i$ from the first layer. 
    \item \emph{Layer 2:} The third layer consists of vertices $w_1,\ldots,w_{(1+\eps)N}$,
                          where each $w_i$ has an edge to~$v_i$ from the second layer.  
    \item \emph{Layer 4:} Let $\Gexp=(L\cup R,E)$ be the expander from Proposition~\ref{prop:alon-expander}.
                          The fourth layer consists of $|L|=(1+\eps)N$ bags, each with at most three vertices.
                          More precisely, if $\dgr(\ell_i)$ is the degree of vertex $\ell_i\in L$ in the expander $\Gexp$,
                          then bag~$L_i$ has $\dgr(\ell_i)$ vertices. Each vertex in $L_i$ has an edge to the vertex~$w_i$
                          from the third layer.
    \item \emph{Layer 5:} Finally, the fifth layer arrives. Each vertex in this layer corresponds to a vertex~$r_i$
                          from the bipartite expander~$\Gexp$ and, with a slight abuse of notation, we will also denote
                          it by~$r_i$. If a vertex $r_i$ in $\Gexp$ has an edge to some vertex $\ell_j$ in $\Gexp$, then the corresponding vertex $r_i$
                          in our construction will have an edge to some vertex in the bag~$L_j$. We do this in
                          such a way that each vertex in any of the bags $L_i$ has an edge to exactly one vertex~$r_i$.
                          In addition to the edges to (vertices in) the bags, each vertex $r_i$ also  has an edge 
                          to the vertex~$v_i$ from the second layer.
\end{itemize}
Let $t_1$ be the time at which the last vertex of $L_{(1+\eps)N}$ was inserted,
and let $t_2$ be the time at which $r_N$ was inserted.
Let $G(t)$ denote the graph induced by all vertices inserted up to time~$t$.
Thus $G(t_1)$ consists of the layers~1--4, and $G(t_2)$ consists of layers~1--5.
\begin{observation}\label{obs:Size of OPT} 
For any $t$ with $t_1\leq t\leq t_2$ we have $\opt(t) \leq (2+2\eps)N$.
Moreover, $\opt(t_1)= (2+2\eps)N$ and $\opt(t_2) \leq (2+\eps)N$.
\end{observation}
\begin{proof}
For any $t_1\leq t\leq t_2$, the set $D_1 := \{v_1,\ldots,v_{(1+\eps)N}\} \cup \{w_1,\ldots,w_{(1+\eps)N} \}$
forms a dominating set for~$G(t)$. It is easily checked that this is optimal for~$G(t_1)$.
Moreover, $D_2 := \{v_1,\ldots,v_{(1+\eps)N}\} \cup \{r_1,\ldots,r_{N} \}$
forms a dominating set for~$G(t_2)$.
\end{proof}
We call a bag $L_i$ \emph{fully dominated} by a set $D$ of vertices if each
vertex in $L_i$ is dominated by some vertex in~$D$. 
Observation~\ref{obs:Size of OPT} states that $\opt(t_2)$ is significantly smaller than $\opt(t_1)$,
which is because the 
vertices in~$R$ can fully dominate all bags. This means that $\Dalg(t_2)$ must contain
most vertices of~$R$, in order to achieve the desired approximation. 
If the stability parameter $f_{\eps^*}(n)$ is small,
however, then at some time $t_1<t<t_2$, the set $\Dalg(t)$ must contain roughly 
$\delta N$ vertices from~$R$. But this  will be expensive, because of the
expander property of $\Gexp$. Hence, if \alg maintains the desired approximation ratio,
then $f_{\eps^*}$ cannot be very small. Next we make this proof idea precise.
\begin{lemma}\label{lemma: Dom too less is bad}
Let $\Dalg(t_2)$ denote the output dominating set for $G(t_2)$. Suppose $\Dalg(t_2) \cap R$ 
fully dominates at most $\delta N$ bags $L_i$. Then $|\Dalg(t_2)| > (1+\eps^*) \cdot \opt(t_2)$.
\end{lemma}
\begin{proof}
Let $m < \delta N$ denote the number of bags fully dominated by $\Dalg(t_2) \cap R$. 
Consider a bag $L_i$ that is not fully dominated by $\Dalg(t_2) \cap R$. 
Then $\Dalg(t_2)$ must contain the vertex $w_i$ or at least one vertex from the bag~$L_i$. 
Hence, the number of vertices in $\Dalg(t_2)$ from the third and fourth layer
is at least $(1+\eps)N-m$. Moreover, $\Dalg(t_2)$ must have at least $(1+\eps)N$ 
vertices from the first and second layer, to dominate all vertices from the first layer. 
Observe that in order to fully dominate a bag $L_i$ by vertices in $R$,
we need all vertices in $R$ with an edge to some vertex of~$L_i$. 
So if $\Dalg(t_2) \cap R$ fully dominates $m\leq \delta N$ bags, then 
$|\Dalg(t_2) \cap R| \geq 1.99m$ by the properties of the expander graph $\Gexp$ in Proposition~\ref{prop:alon-expander},
with $\mu=0.005$. 
Hence,
\[
|\Dalg(t_2)| > (1+\eps)N-m+(1+\eps)N+ 1.99m \geq (2+2\eps)N.
\]
Observation~\ref{obs:Size of OPT} thus implies that
$\frac{|\Dalg(t_2)|}{\opt(t_2)} \geq \frac{2+2\eps}{2+\eps} > 1+\eps^*$.
\end{proof}
Lemma~\ref{lemma: Dom too less is bad} means that, in order to achieve approximation ratio $1+\eps^*$,
the set $\Dalg(t_2) \cap R$ must fully dominate more than $\delta N$ bags. 
Next we show that this cannot be done when the stability parameter $f_{\eps^*}(n)$
is~$o(n)$.
\begin{lemma}\label{lemma:Dom too much is also bad}
Let $t^*$ be the first time when $\Dalg(t^*) \cap R$ fully dominates at least $\delta N$ bags. 
If $f_{\eps^*}(n) < \frac{1}{6(7+6\eps)}\lfloor{\delta n}\rfloor$
then $|\Dalg(t^*)| > (1+\eps^*)\cdot \opt(t^*)$. 
\end{lemma}
\begin{proof}
Let $n_{t}$ denote the number of vertices of the graph~$G(t)$, and observe that 
$n_{t^*} \leq (7+6\eps)N$.  Hence, $f_{\eps^*}(n_{t^*})\leq \frac{1}{6}\delta N$. 
By definition of $t^*$, we know that just before time~$t^*$ the set $\Dalg\cap R$
fully dominates less than $\delta N$ bags. Because \alg is $f_{\eps^*}(n)$-stable,
the number of vertices from $R$ added to $\Dalg$ at time $t^*$ is at most $f_{\eps^*}(n_{t^*})$. Since these
new vertices have degree at most three, they can complete the full domination of at most
$3 f_{\eps^*}(n_{t^*})$ bags. Thus, 
\[
\big(\;\mbox{number of bags fully dominated by $\Dalg(t^*) \cap R$}\;\big)  
\ \ < \ \ \delta N + 3 f_{\eps^*}(n_{t^*}) 
\ \ \leq \ \ \frac{3}{2}\delta N. 
\]
Let $L_i$ be a bag that is not fully dominated by $\Dalg(t^*) \cap R$. 
Since $\Dalg(t^*)$ is a dominating set, it must then contain the vertex $w_i$ 
or at least one vertex from~$L_i$. Hence, the number of vertices in $\Dalg(t^*)$ from layers~3 
and~4 is more than $(1+\eps)N- \frac{3}{2}\delta N$. In addition, $\Dalg(t^*)$ must have at least 
$(1+\eps)N$ vertices from layers~1 and~2. Finally, in order to fully dominate a bag $L_i$ by vertices in $R$, 
we need all the vertices in $R$ that have an edge to some vertex of $L_i$. In other words,
$\Dalg(t^*)\cap R$ must contain all neighbors of the fully dominated bags.
Since $\Dalg(t^*) \cap R$ dominates at least $\delta N$ bags,
we know that $|\Dalg(t^*) \cap R| \geq 1.99 \cdot \delta N$,
by the properties of the expander graph in Proposition~\ref{prop:alon-expander}. 
Hence,
\[
\begin{array}{lll}
|\Dalg(t^*)| & >  & (1+\eps)N - \frac{3}{2}\delta N + (1+\eps)N+ 1.99\cdot \delta N \\
             & \geq & 
                       \left(1+ \frac{0.49 \delta}{2+2\eps}\right) (2+2\eps)N \\
             & \geq  &(1+\eps^*)\cdot \opt(t^*).
\end{array}
\]
\end{proof}
By Lemmas \ref{lemma: Dom too less is bad} and \ref{lemma:Dom too much is also bad} 
we obtain the following result.
\begin{theorem}
There is a constant $\eps^*>0$ such that any dynamic $(1+\eps^*)$-approximation algorithm 
for \domset in the vertex arrival model, must have stability parameter~$\Omega(n)$,
even when the maximum degree of any of the graphs $G(t)$ is bounded by~4.
\end{theorem}

\subsection{Constant-stability algorithms when the arrival degrees are bounded}
In the previous section we saw that there is no SAS for \domset, even when the
maximum degree is bounded by~4. In this section we present three stable algorithms
whose approximation ratio depends on the maximum arrival degree, $d$, of the vertices: 
a simple 1-stable algorithm with approximation ratio $(d+1)^2$,
a more complicated 3-stable algorithm with approximation ratio~$9d/2$,
and an $O(d)$-stable algorithm with approximation ratio~$O(1)$.
Note that we only restrict the degree upon arrival:
the degree of a vertex may further increase due to the arrival of new vertices.
This implies that deletions cannot be handled by a stable algorithm with
bounded approximation ratio: even with arrival degree~1 we may create a star graph
of arbitrarily large size, and deleting the center of the star cannot be handled
in a stable manner without compromising the approximation ratio.

\subparagraph*{A 1-stable $(d+1)^2$-approximation algorithm.}
Recall that $G(t)$ denotes the graph after the arrival of the $t$-th vertex.
We can turn $G(t)$ into a directed graph $\vec{G}(t)$, by directing each edge
towards the older of its two incident vertices. In other words, when a new
vertex arrives then its incident edges are directed away from it.

Let $\Nout[v] := \{v\} \cup \{\text{out-neighbors of $v$ in $\vec{G}(t)$}\}$, where~$t$ is such
that $v$ is inserted at time~$t$. In other words, $\Nout[v]$ contains $v$ itself plus the neighbors of $v$ 
immediately after its arrival.
Let $\optout(t)$ denote the minimum size of a dominating set in $\vec{G}(t)$ under the condition 
that every vertex~$v$ should be dominated by a vertex in $\Nout[v]$. We call such a dominating
set a \emph{directed dominating set}. Note that a directed dominating set for $\vec{G}(t)$
is a dominating set in $G(t)$ as well. The following lemma states that $\optout(t)$
is not much larger than $\opt(t)$.
\begin{lemma} \label{size of opt_out}
At any time $t$ we have $\optout(t) \leq (d+1) \cdot \opt(t)$.
\end{lemma}
\begin{proof} 
Let $\Dopt(t)$ be a minimum dominating set for $G(t)$. Let $D:=\bigcup_{v \in \Dopt(t)} \Nout[v]$. 
Observe that $D$ is a directed dominating set for $\vec{G}(t)$. Since every vertex arrives with degree at most $d$,
we have $|\Nout[v]|=d+1$. The result follows.
\end{proof}
We call two vertices $u,v$ \emph{unrelated} if $\Nout[u] \cap \Nout[v] = \emptyset$, otherwise $u,v$ are \emph{related}. The following
lemma follows immediately from the definition of $\optout(t)$.
\begin{lemma} \label{size of U}
Let $U(t)$ be a set of pairwise unrelated vertices in $\vec{G}(t)$. Then $\optout(t) \geq |U(t)|$.
\end{lemma}
Our algorithm will maintain a directed dominating set $\Dalg(t)$
for $\vec{G}(t)$ and a set $U(t)$ of pairwise unrelated vertices. Since the initial graph is empty,
we initialize $\Dalg(0) := \emptyset$ and $U(0) := \emptyset$. When a new vertex $v$ arrives
at time $t$, we proceed as follows.
\begin{algorithm}[H] 
\caption{\textsc{Directed-DomSet}($v$)}
\begin{algorithmic}[1]
\State $\rhd$ $v$ is the vertex arriving at time~$t$
\If {$\Nout[v] \cap \Dalg(t-1) \neq \emptyset$} \hfill $\rhd$ $v$ is already dominated 
       \State Set $\Dalg(t) \gets \Dalg(t-1)$ and $U(t) \gets U(t-1)$
       \Else
       \If {$v$ is unrelated to all vertices $u\in U(t-1)$}
          \State Set $U(t) \gets U(t-1) \cup \{v\}$ and $\Dalg(t) \gets \Dalg(t-1) \cup \{v\}$
       \Else 
          \State Let $u$ be a vertex related to $v$, that is, with $\Nout[u] \cap \Nout[v] \neq \emptyset$.
          \State Pick an arbitrary vertex $w \in \Nout[u] \cap \Nout[v]$.
          \State Set $\Dalg(t) \gets \Dalg(t-1) \cup \{w\}$ and $U(t) \gets U(t-1)$.
       \EndIf
\EndIf
\end{algorithmic}
\end{algorithm}
This leads to the following theorem.
\begin{theorem}
There is a $1$-stable $(d+1)^2$-approximation algorithm for \domset in the vertex-arrival model, 
where $d$ is the maximum arrival degree of any vertex.
\end{theorem}
\begin{proof}
Clearly the set $\Dalg$ maintained by \textsc{Directed-DomSet} is a directed dominating set.
Moreover the algorithm is $1$-stable as after each arrival, it either adds 
a single vertex to $\Dalg$ or does nothing. It is easily checked that the set $U$ maintained 
by the algorithm is always a set of pairwise unrelated vertices, and that all vertices in $\Dalg$ 
are an out-neighbor of a vertex in $U$ or are in $U$ themselves. Hence, by Lemmas~\ref{size of opt_out} 
and~\ref{size of U}, at any time~$t$ we have
\[
|\Dalg(t)| \leq (d+1) \cdot |U(t)| \leq (d+1) \cdot \optout(t) \leq (d+1)^2 \cdot \opt(t),
\]
which finishes the proof.
\end{proof}

We now show that the approximation ratio $\Theta(d^2)$ is tight for \textsc{Directed-DomSet}. Indeed, the instance depicted in Fig.~\ref{fig:tight-example} shows that \textsc{Directed-DomSet} has approximation 
ratio $\Omega(d^2)$. First, vertex $a_1$ arrives, which will be put  into $U$ and~$\Dalg$.
Then $a_2$ and $v_1,\ldots,v_{d^2}$ arrive, with an outgoing edge to~$a_1$. These
vertices are dominated by~$a_1$, so $U$ and $\Dalg$ are not modified.
Next, $w_1,\ldots,w_d$ arrive.
Vertex $w_1$ arrives with neighbors $v_1,\ldots,v_d$. We have $U=\Dalg=\{a_1\}$,
so $w_1$ is not dominated by any node in $\Dalg$ and it is not related to any node in $\Dalg$. 
Hence, $w_1$ will be put  into $U$ and~$\Dalg$. 
Then vertex $w_2$ arrives with neighbors $v_{d+1},\ldots,v_{2d}$. We have $U=\Dalg=\{a_1,w_1\}$,
so $w_1$ is not dominated by any node in $\Dalg$ and it is not related to any node in $\Dalg$. 
Hence, $w_2$ will be put  into $U$ and~$\Dalg$. Similarly, $w_3,\ldots,w_d$
will all be put into $U$ and $\Dalg$.
Then $x_1,\ldots,x_{d(d-1)}$ arrive, each with one outgoing edge, to a unique node
among $v_1,\ldots,v_{d^2}$, as shown in the figure. They are all related to a node in $U$,
namely one of the $w_i$'s, so their out-neighbors will be put into $\Dalg$.
Finally, $z$ arrives, which is also put into $\Dalg$.
Hence, at the end of the algorithm $|\Dalg|=d^2+2$, while $\opt=3$ since $\{a_1,a_2,z\}$
is a dominating set.

\begin{figure}
\begin{center}
\includegraphics{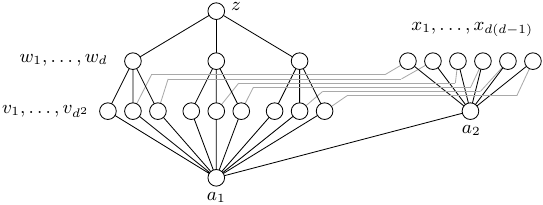}
\end{center}
\caption{Example showing that \textsc{Directed-DomSet} has approximation ratio $\Omega(d^2)$.}
\label{fig:tight-example}
\end{figure}


\subparagraph*{A 3-stable $(9d/2)$-approximation algorithm.}
Algorithm \textsc{Directed-DomSet} has lowest possible stability, but its approximation ratio
is~$\Omega(d^2)$. We now present an algorithm whose stability is still very small,
namely~3, but whose approximation ratio is only~$O(d)$. This is asymptotically optimal,
since any algorithm with constant stability must have\footnote{This can be seen by considering the following scenario: first vertices $v_1,\ldots,v_d$
arrive as singletons, and then vertex $v_{d+1}$ arrives with edges to all vertices $v_i$ with $i\leq d$. Any $O(1)$-stable algorithm
will then have approximation ratio at least~$d-O(1)=\Omega(d)$.} 
approximation ratio~$\Omega(d)$.
Our approach is somewhat similar to that of Liu and Toole-Charignon~\cite{10.1007/978-3-031-18367-6_7},
but a key difference is that we obtain a worst-case bound on the stability and they obtain an amortized bound.
Our algorithm works in \emph{phases}, as explained next. Suppose we start a new phase at time~$t$
and let $\Dalg(t-1)$ be the output dominating set at time $t-1$. The algorithm computes a minimum
dominating set $\Dopt(t)$ for the graph~$G(t)$, which we call the
\emph{target dominating set}. The algorithm will then slowly migrate from
$\Dalg(t-1)$ to $\Dopt(t)$, by first adding the vertices in $D^+(t) := \Dopt(t)\setminus \Dalg(t-1)$ 
and then removing the vertices in $D^-(t) := \Dalg(t-1) \setminus \Dopt(t)$. 
This is done in $\left\lceil |D^+(t) \cup D^-(t)|/2 \right\rceil$ steps.
Vertices that arrive in the meantime are also added to the dominating
set, to ensure that the output remains a dominating set at all times.
After all vertices in $D^+(t)$ and $D^-(t)$ have been added and deleted, respectively,
the next phase starts. Next we describe and analyze the algorithm in detail.

At the start of the whole algorithm, at time $t=0$, we initialize 
$D_{\mathrm{alg}}(0):=\emptyset$, and  $D(0)^+=\emptyset$ and  $D(0)^-=\emptyset$.
\begin{algorithm}[H] 
\caption{\textsc{Set-and-Achieve-Target-DS}($v$)}
\begin{algorithmic}[1]
\State $\rhd$ $v$ is the vertex arriving at time $t$ and $G(t)$ is the graph after arrival of~$v$
\State $\Dalg \gets \Dalg(t-1) \cup \{v\}$ \label{step:add-new-vertex}
\If{$D^+(t-1)=\emptyset$ and $D^-(t-1)=\emptyset$} \hfill $\rhd$ start a new phase
    \State Let $\Dopt(t)$ be a minimum dominating set for $G(t)$. 
    \State        Set $D^+(t) \gets \Dopt(t) \setminus \Dalg(t-1)$ and $D^-(t) \gets \Dalg(t-1) \setminus \Dopt(t)$. 
\Else 
    \State Set $D^+(t) \gets D^+(t-1)$ and $D^-(t) \gets D^-(t-1)$
\EndIf
\State Set $m^+ \gets \min (2,|D^+(t)|)$. Delete $m^+$ vertices from $D^+(t)$ and add them to $\Dalg$. \label{step:d-}
\State Set $m^- \gets \min(2 - m^+, |D^-(t)|)$. Delete $m^-$ vertices from $D^-(t)$ and delete the same vertices from $\Dalg$. \label{step:d+}
\State $\Dalg(t) \gets \Dalg$
\end{algorithmic}
\end{algorithm}
The algorithm defined above is $3$-stable, as it adds one vertex to $\Dalg$ in 
step~\ref{step:add-new-vertex} and then makes two more changes to $\Dalg$ in steps~\ref{step:d-} and~\ref{step:d+}.
Next we prove that its approximation ratio is bounded by $9d/2$. Note that the size of a minimum dominating set can
reduce over time, due to the arrival of new vertices. The next lemma shows that this
reduction is bounded. Let $\maxopt(t) := \max( \opt(1), \opt(2), \ldots , \opt(t))$ denote the maximum size 
of any of the optimal solutions until (and including) time~$t$.
\begin{lemma}\label{lem:maxopt}
For any time $t$ we have $\maxopt(t)\leq d \cdot \opt(t)$, where $d$ is 
the maximum arrival degree of any vertex.
\end{lemma}
\begin{proof} 
Let $t^*\leq t$ be such that $\maxopt(t)=\opt(t^*)$.
Let $\Dopt(t)$ be an optimal dominating set for $G(t)$ and define $V(t^*)$ to be the set 
of vertices of $G(t^*)$. Let $D$ be the set of vertices in $\Dopt(t)$ that were not yet 
present at time~$t^*$, and define $D^* := \left( \Dopt(t) \setminus D \right) \cup N_{t^*}(D)$,
where $N_{t^*}(D)$ contains the neighbors of $D$ in $V(t^*)$.
Then $D^*$ is a dominating set for $G(t^*)$ since any vertex in $V(t^*)$ that is not dominated
by a vertex in  $\Dopt(t)\setminus D$ is in $D^*$ itself. Moreover, $|D^*| \leq d \cdot\opt(t)$,
since each vertex in $D$ has at most $d$ neighbors in $V(t^*)$.
%
%
\end{proof}
We first bound the size of $\Dalg$ at the start of each phase. 
Note that in the proofs below, $D^+(t)$ and $D^-(t)$ refer to the situation
before the execution of line $8$ and $9$ in the algorithm \textsc{set-and-achieve-target-DS}.
\begin{lemma}\label{cardinality of output at set phase}
If a new phase starts at time $t$, then $|\Dalg(t-1)| \leq 3\cdot \maxopt(t-1)$.
\end{lemma}
\begin{proof}
     We proceed by induction on~$t$. The lemma trivially holds at the start of the first
     phase, when $t=1$. Now consider the start of some later phase, at time~$t$,
     and let $\tprev<t$ be the previous time at which a new phase started.
  Recall that 
     \[
     \Dalg(t-1) = \Dopt(\tprev) \cup \{\mbox{vertices arriving at times $\tprev,\tprev+1,\ldots,t-1$}\}. 
     \]
     Moreover,  
     \[
     |D^+(\tprev) \cup D^-(\tprev)| \leq |\Dalg(\tprev-1)| + |\Dopt(\tprev)| \leq 3\cdot \maxopt(\tprev-1) + \opt(\tprev),
     \]
     where the last inequality uses the induction hypothesis. From time $\tprev$ up to time $t-1$,
     the vertices from $D^+(\tprev) \cup D^-(\tprev)$ are added/deleted in pairs, so
     \[
     t-\tprev = \left\lceil \frac{3\cdot\maxopt(\tprev-1)+\opt(\tprev)}{2} \right\rceil 
            \leq \left\lceil \frac{4\cdot\maxopt(\tprev)}{2} \right\rceil  
              = 2\cdot\maxopt(\tprev).
     \]
     Hence, 
     \[
     \begin{array}{lll}
     |\Dalg(t-1)| & \leq & \opt(\tprev)+ (t-\tprev) \\
       & \leq & \maxopt(\tprev)+ 2\cdot \maxopt(\tprev) \\
       & \leq & 3\cdot \maxopt(t-1).

     \end{array}
     \]
\end{proof}
The previous lemma bounds $|\Dalg|$ just before the start of each phase. Next we use this
to bound $|\Dalg|$ during each phase.
\begin{lemma} \label{lem:during-phase}
For any time $t$ we have $|\Dalg(t)| \leq (9/2)\cdot \maxopt(t)$.
\end{lemma}
\begin{proof}

Consider a time~$t$. If a new phase starts at time~$t+1$ then the lemma follows directly
from Lemma~\ref{cardinality of output at set phase}. Otherwise, let $\tprev\leq t$ be the 
previous time at which a new phase started, and let $\tnext>t$ be the next time
at which a new phase starts. Furthermore, 
let $t^* := \max\{t': \tprev \leq t'< \tnext \mbox{ and } D^+(t')\neq\emptyset \}$.
In other words, $t^*$ is the last time step in the interval $[\tprev,\tnext)$ 
at which we still add vertices from $D^+$ to~$\Dalg$. If $D^+(\tprev)$ is empty then let $t^*=\tprev$. 
It is easy to see from the algorithm that $|\Dalg(t)| \leq |\Dalg(t^*)|$.
Note that $\Dalg(t^*)$ contains the vertices from $\Dalg(\tprev-1)$, plus the
vertices from $\Dopt(\tprev)$, plus the vertices that arrived from
time $\tprev$ to time~$t^*$. Hence,
\[
\begin{array}{lll}
|D_{\mathrm{alg}}(t^*)| & \leq & |D_{\mathrm{alg}}(\tprev-1)| + \opt(\tprev)+ (t^* -\tprev +1) \\
     & \leq & |\Dalg(\tprev-1)| + \opt(\tprev)+ \frac{\opt(\tprev)}{2}  \\
     & \leq  & 3 \cdot \maxopt(\tprev-1)+ \maxopt(\tprev) + \frac{\opt(\tprev)}{2} \\
     & \leq & (9/2)\cdot \maxopt(t).
\end{array}
\]
Note that from the first to the second line we replaced $(t^*-\tprev+1)$ by $\opt(\tprev)/2$,
which we can do because we add vertices from $\Dopt(\tprev)$ in pairs. It may seem that we
should actually write $\left\lceil \opt(\tprev)/2 \right\rceil$ here. When $\opt(\tprev)$
is odd, however, then the algorithm can already remove\footnote{This is not true in the special case when $D^-(\tprev) = \emptyset$,
but in that case the second term is an over-estimation.
Indeed, when $D^-(\tprev) = \emptyset$ then $\Dalg(\tprev-1) \subseteq \Dopt(\tprev)$
and so $|D^+(\tprev)| = |\Dopt(\tprev) \setminus \Dalg(\tprev-1)| < |\opt(\tprev-1)|$.} 
a vertex in $D^-$ from $\Dalg$  when the last vertex from $D^+$
is added to $\Dalg$, which compensates for the omission of the ceiling function.
This finishes the proof. 
\end{proof}
Putting Lemmas~\ref{lem:maxopt} and~\ref{lem:during-phase} together, we obtain the following theorem.
\begin{theorem}
There is a $3$-stable $(9d/2)$-approximation algorithm for \domset in the vertex-arrival model, 
where $d$ is the maximum arrival degree of any vertex.
\end{theorem}


%% file: constant_approx_domset.tex
\subparagraph*{An $O(d)$-stable $O(1)$-approximation algorithm.}
The technique of the algorithm given above can also be used to obtain other
trade-offs between the stability and the approximation ratio. This is done
by migrating more aggressively during each phase: instead of adding/deleting
a pair of vertices in each step, we add/delete $k$ vertices, for some parameter $k\geq 2$,
which results in a $(k+1)$-stable algorithm. Thus, in \textsc{set-and-achieve-target}
we replace the value~2 in steps~\ref{step:d-} and~\ref{step:d+} by~$k$.
Here we analyze the variant where we take $k$ large enough to get an $O(1)$-approximation.
For this we need to work with a different lower bound than the one obtained from Lemma~\ref{lem:maxopt}.
This new lower bound is derived from the following observation.
\begin{observation}\label{obs:opt-decrease}
For any time $t$ we have $\opt(t)-(d-1) \leq \opt(t+1)\leq \opt(t)+1$, where $d$ is 
the maximum arrival degree of any vertex.
\end{observation}
Now consider \textsc{Set-and-Achieve-Target} where we set $k := 22d+1$,
thus obtaining a $(22d+2)$-stable algorithm. We can now show that at the start of each phase, $|\Dalg(t-1)|$ 
is bounded in terms by $\opt(t)$ instead of $\maxopt$.
\begin{lemma}\label{lem:start-phase-22d+1}
If a phase starts at time $t$, then we have $|\Dalg(t-1)| \leq (9/2)\cdot \opt(t-1)$.
\end{lemma}
\begin{proof}
     The proof is similar to the proof of Lemma~\ref{cardinality of output at set phase}. 
     This time the induction hypothesis and the fact that $\opt(\tprev) \leq \opt(\tprev-1)+1$ gives us
     \[
     |D^+(\tprev) \cup D^-(\tprev)| \leq |\Dalg(\tprev-1)| + |\Dopt(\tprev)| \leq \frac{11}{2}\cdot \opt(\tprev-1)+1,
     \]
     Now, if $\opt(\tprev-1) \leq 4d$ then we have  
     \[
     |D^+(\tprev) \cup D^-(\tprev)| \leq \frac{11}{2}\cdot \opt(\tprev-1) + 1\leq 22d+1.
     \]
     But then we must have $t=\tprev+1$, because we can add $D^+(\tprev)$
     and remove $D^-(\tprev)$ in a single step. Thus $\Dalg(t-1)=\opt(\tprev)=\opt(t-1)$, 
     and so our lemma holds. 
     Hence, we can assume that $\opt(\tprev-1) > 4d$. By Observation~\ref{obs:opt-decrease} we then also have $\opt(\tprev) > 3d$.
     
     Now from time $\tprev$ up to time $t-1$,
     the vertices from $D^+(\tprev) \cup D^-(\tprev)$ are added/deleted in groups of size $22d+1$.
     Hence,
     \[
     t- \tprev
            \leq \left\lceil \frac{\frac{11}{2}\cdot\opt(\tprev-1)+1}{22d+1} \right\rceil  
              \leq \frac{\opt(\tprev-1)}{4d} +1 < \frac{\opt(\tprev-1)}{2d},
     \]
     where the last inequality uses that $\opt(\tprev-1) > 4d$. 
     We also have
     \[
     \frac{\opt(\tprev -1)}{d} \leq \opt(\tprev -1) - (d-1) \leq \opt(\tprev),
     \]
     where the first inequality uses that $\opt(\tprev-1) > 4d$ and the second inequality uses Observation~\ref{obs:opt-decrease}.
     Hence,
     \[
     D_{\mathrm{alg}}(t-1)  \leq  \opt(\tprev)+ (t-\tprev) 
        \leq  \opt(\tprev)+ \frac{\opt(\tprev-1)}{2d}  \leq
         \frac{3}{2}\cdot \opt(\tprev). 
     \]
     Also, by Observation~\ref{obs:opt-decrease} and using $\opt(\tprev) > 3d$ we have
     \[
     \begin{array}{lll}
     \opt(t-1) & \geq & \opt(\tprev) - (d-1) (t-\tprev) \\[2mm]
               & \geq & \opt(\tprev) - (d-1) \cdot \frac{\opt(\tprev-1)}{2d}  \\[2mm]
              & \geq & \opt(\tprev) - (d-1) \cdot \frac{\opt(\tprev)+(d-1)}{2d} \\[2mm]
               & \geq & \opt(\tprev) - d \cdot \frac{\opt(\tprev)+(d-1)}{2d} \\[2mm]
              & = & \frac{\opt(\tprev)}{2} - \frac{(d-1)}{2}  \\[2mm]
              & \geq & \frac{\opt(\tprev)}{2} - \frac{\opt(\tprev)}{6} \\[2mm]
               & = & \frac{\opt(\tprev)}{3}.
     \end{array}
     \]
    We conclude that $|\Dalg(t-1)| \leq (9/2)\cdot \opt(t-1)$, thus finishing the proof.
\end{proof}
We can now bound the approximation ratio of our algorithm.
\begin{lemma} \label{lem:during-phase 22d+1}
For any time $t$ we have $|\Dalg(t)| \leq (45/2)\cdot \opt(t)$.
\end{lemma}
\begin{proof}
The proof is similar to the proof of Lemma~\ref{lem:during-phase}.
Consider a time~$t$. If a new phase starts at time~$t+1$ then the lemma follows directly
from Lemma~\ref{lem:start-phase-22d+1}. Otherwise, let $\tprev\leq t$ be the 
previous time at which a new phase started and let $\tnext>t$ be the next time
at which a new phase starts. Furthermore, 
let $t^* := \max\{t': \tprev \leq t'< \tnext \mbox{ and } D^+(t')\neq\emptyset \}$.  If $D^+(\tprev)$ is empty then let $t^*=\tprev$. 
In other words, $t^*$ is the last time step in the interval $[\tprev,\tnext)$ 
at which we still add vertices from $D^+$ to~$\Dalg$.
It is easy to see from the algorithm that $|\Dalg(t)| \leq |\Dalg(t^*)|$.
Note that $\Dalg(t^*)$ contains the vertices from $\Dalg(\tprev-1)$, plus the
vertices from $\Dopt(\tprev)$, plus the vertices that arrived from
time $\tprev$ to time~$t^*$.

If $\opt(\tprev-1) \leq 4d$ then, as observed in the proof of Lemma~\ref{lem:start-phase-22d+1},
the phase that started at $\tprev$ must immediately end at $\tprev$, and so $t=\tprev$. 
Hence, $\Dalg(t)=\Dalg(\tprev)=\opt(\tprev)$ and our lemma holds. Thus we can assume $\opt(\tprev-1) > 4d$.

As in the proof of Lemma \ref{lem:start-phase-22d+1} we also have  
\[
    t^*+1 -\tprev \leq \tnext-\tprev 
            <  \frac{\opt(\tprev-1)}{2d}.
\]
Also from the proof of Lemma \ref{lem:start-phase-22d+1}, we have for all $t \in [\tprev, \tnext)$ we have, 
\[
\opt(t) \geq \opt(\tprev) - (d-1) (t-\tprev) \geq  \opt(\tprev) - (d-1) \left( \frac{\opt(\tprev-1)}{2d} \right) \geq \frac{\opt(\tprev)}{3},
\]
Hence,
\[
\begin{array}{lll}
|D_{\mathrm{alg}}(t^*)| & \leq & |D_{\mathrm{alg}}(\tprev-1)| + \opt(\tprev)+ (t^* -\tprev +1) \\[2mm]
     & \leq & |\Dalg(\tprev-1)| + \opt(\tprev)+ \frac{\opt(\tprev-1)}{2d}  \\[2mm]
     & \leq  & \frac{9}{2} \cdot \opt(\tprev-1)+ \opt(\tprev) + \frac{\opt(\tprev-1)}{2d} \\[2mm]
 & \leq  & \frac{9}{2} \cdot \opt(\tprev)+ \frac{9}{2} \cdot (d-1)+ \frac{3}{2}\opt(\tprev)  \\[2mm]
& \leq  & \frac{9}{2} \cdot \opt(\tprev)+ \frac{3}{2}\opt(\tprev) + \frac{3}{2}\opt(\tprev)  \\[2mm]
     & \leq  & \frac{15}{2}\cdot \opt(\tprev).
\end{array}
\]

This finishes the proof.
\end{proof}
Putting Lemmas~\ref{lem:start-phase-22d+1} and~\ref{lem:during-phase 22d+1} together, 
we obtain the following theorem.
\begin{theorem}
There is an $O(d)$-stable $O(1)$-approximation algorithm for \domset in the vertex-arrival model, 
where $d$ is the maximum arrival degree of any vertex.
\end{theorem}

%% file: Independent_Set.tex
\section{Stable Approximation Algorithms for \mis}
In this section we first show 
that \mis does not admit a SAS, even when restricted to graphs of maximum degree~3. Then we give
an $2$-stable $O(d)$-approximation algorithm for graphs whose average degree is bounded by~$d$.

\subsection{No SAS for graphs of maximum degree~3} \label{subsec:no-SAS-mis}
We prove our no-SAS result for \mis in a similar (but simpler) way as for \domset. 
Thus we actually prove the stronger result that there is a constant $\eps^*>0$ 
such that any dynamic $(1+\eps^*)$-approximation algorithm 
for \mis in the vertex arrival model, must have stability parameter~$\Omega(n)$,
in this case even when the maximum degree of any of the graphs $G(t)$ is bounded by~3.

Let $\eps^*>0$ be a real number less than $\min \left(\frac{0.82\delta}{1-0.82\delta}, \eps \right)$. 
Let \alg be an algorithm that maintains an independent set~$\Ialg$ such that
$\opt(t) \leq (1+\eps^*)\cdot|\Ialg(t)|$ at all times.
Let $f_{\eps^*}(n)$ denote the stability of \alg, that is, the maximum number of changes it performs on $\Ialg$ when 
a new vertex arrives, where $n$ is the number of vertices before the arrival. 
We will show that, for arbitrarily large $n$, there is a
sequence of $n$ arrivals that requires $f_{\eps^*}(n) \geq \frac{1}{6(2+\eps)}\lfloor{\delta n}\rfloor$.
As before, choose $N$ large enough such that the bipartite expander 
graph $\Gexp=(L\cup R,E)$ from Proposition~\ref{prop:alon-expander} 
exists for $\mu=0.005$ and $|R|=N$. In our no-SAS construction for \mis,
we only need $\Gexp$, no additional layers are needed. 
Thus the construction is simply as follows.
\begin{itemize}
    \item First the vertices $r_1,\ldots, r_{N}$ from the set $R$ arrive one by one, as singletons.
    \item Next the vertices $\ell_1,\ldots, \ell_{(1+\eps)N}$  from $L$ arrive one by one (in any order),
          along with their incident edges in $\Gexp$.
\end{itemize}
\begin{lemma} \label{thm:property of construction}
Let $t^*$ be the first time when $|\Ialg(t)\cap L| \geq \delta N$.
If $f_{\eps^*}(n) < \frac{1}{6(2+\eps)}\lfloor{\delta n}\rfloor$
then $\opt(t^*) > (1+\eps^*)\cdot |\Ialg(t^*)|$.
\end{lemma}
\begin{proof}
Let $n_{t}$ denote the number of vertices of the graph~$G(t)$, and observe that 
$n_{t^*} \leq (2+\eps)N$.  Hence, $f_{\eps^*}(n_{t^*})\leq \frac{1}{6}\delta N$. 
By definition of $t^*$, we know that just before time~$t^*$ we have $|\Ialg\cap L|< \delta N$. 
Because \alg is $f_{\eps^*}(n)$-stable, we have 
\[
|\Ialg(t^*) \cap L| \leq \delta N + f(n_{t^*}) \leq \frac{7}{6}\delta N.
\]
Since $\mu=0.005$, from Proposition~\ref{prop:alon-expander} we have 
$|N(\Ialg(t^*) \cap L|\geq 1.99 \cdot \delta N$. Hence $|\Ialg(t^*) \cap R| \leq N- 1.99 \cdot \delta N$,
and so
\[
|\Ialg(t^*)|=|\Ialg(t^*) \cap R| + |\Ialg(t^*) \cap L| \leq \frac{7}{6}\delta N + N- 1.99 \cdot \delta N< N- 0.82\,\delta N. 
\]
We also have $\opt(t^*)\geq N$. Hence, $\opt(t^*)>\frac{N}{N- 0.82\delta N} |\Ialg(t^*)|>(1+\eps^*)\cdot |\Ialg(t^*)|$.
\end{proof}

\begin{theorem}
There is a constant $\eps^*>0$ such that any dynamic $(1+\eps^*)$-approximation algorithm 
for \mis in the vertex arrival model, must have stability parameter~$\Omega(n)$,
even when the maximum degree of any of the graphs $G(t)$ is bounded by~3.
\end{theorem}
\begin{proof}
By Proposition~\ref{prop:alon-expander} the maximum degree of the graph is always bounded by $3$. 
Let $t=2N+\eps N$ and let $m := |\Ialg(t) \cap L|$. We know by Lemma~\ref{thm:property of construction} 
that if \alg is a $(1+\eps^*)$-approximation and if 
$f_{\eps^*}(n) < \frac{1}{6(2+\eps)}\lfloor{\delta n}\rfloor$ then $m \leq \delta n$. Hence $|\Ialg(t) \cap R| \leq N-1.99m$. So we have 
\[
|\Ialg(t)|=|\Ialg(t) \cap R| + |\Ialg(t) \cap L| \leq N-1.99m+ m \leq N.
\]
But $\opt(t)= (1+\eps)N$. Hence, the approximation ratio at time $t=(2+ \eps)N$ is at least
$\frac{(1+\eps)N}{N}= 1+\eps>1+\eps^*$, which is a contradiction.
This finishes the proof.
\end{proof}

\subsection{Constant-stability algorithms when the average degree is bounded}
In this section we consider the setting where the average degree of the graphs~$G(t)$ is bounded 
by some constant~$d$. We first give an algorithm for the vertex-arrival
model where only insertions are allowed, and then we extend our algorithm so that it can
handle deletions as well.

\subparagraph*{A 2-stable $O(d)$-approximation algorithm in the insertion-only setting.}
If we allow just one change after each vertex arrival, then it is impossible to get a 
bounded approximation ratio.\footnote{After the arrival of the first vertex~$v_1$,
we can create an arbitrarily large star graph with $v_1$ as center. With a stability of~1, 
the algorithm can never remove $v_1$ from its independent set~$\Ialg$, since that would
lead to an empty independent set. But then it will not be able to add any 
other vertex of the star to $\Ialg$, leading to an arbitrarily bad approximation ratio.}
However, we can get a bounded approximation ratio with only two changes per arrival.

Observe that if the maximum degree (rather than average degree)
is bounded by some constant $d^*$,
then a simple greedy $1$-stable algorithm maintains an $O(d^*)$-approximation. 
Our idea is to maintain an induced subgraph with a number of vertices that is linear in the 
number of vertices of $G(t)$, and whose maximum degree is bounded.
We then use the induced subgraph to generate an independent set. Below we make the idea precise.

We will need the following (trivial) subroutine, which takes as input a graph
$G=(V,E)$, an independent set~$I^*$ in $G$, and a subset $W^*\subset V$, 
and tries to add a vertex $v$ from $W^* \setminus I^*$ 
to $I^*$ such that $I^* \cup \{v\}$ is still an independent set. 
\begin{algorithm}[H] 
\caption{\textsc{Greedy-Addition}($G,I^*,W^*$)}
\begin{algorithmic}[1]
\If {there exist a vertex $v\in W^* \setminus I^*$ such that $I^* \cup \{v\}$ is an independent set in $G$} 
       \State Pick an arbitrary such vertex~$v$ and set $I^*= I^* \cup \{v\} $
       \EndIf
\end{algorithmic}
\end{algorithm}
Next we describe our 2-stable algorithm. 
Let $\Delta(G)$ denote the maximum degree of a graph $G=(V,E)$ and
define $G[W]$ to be the subgraph of $G$ induced by a subset~$W\subset V$. 
Let $V(t)$ denote the set of vertices of~$G(t)$.

Observe that by ordering the vertices of $G(t)$ in increasing order of their degree and 
by taking  the first $\frac{99}{100}|V(t)|$ vertices,
we can construct a set $V^*(t)\subseteq V$ such that\footnote{More precisely, we have
$|V^*(t)| = \floor{\frac{99}{100}|V(t)|}$. For convenience we ignore this rounding issue. It is easily checked
that this does not influence our final (asymptotic) bound on the approximation ratio. }
$|V^*(t)|= \frac{99}{100}|V(t)|$ 
and $\Delta(G[V^*(t)]) \leq 100d$. (The number 100 has no special significance---it can
be chosen much smaller---but we use it for convenience.)
The idea of our algorithm is to maintain a vertex set $W(t)\subseteq V(t)$ 
such that $\Delta(G[W(t)])\leq 100d$ and the size of $W(t)$ is linear in $|V(t)|$. 
In order to maintain such a subset we work in phases, as before.
At the start of each phase, the algorithm sets itself a target vertex set $V^*(t)$ 
of large size and with $\Delta(G[V^*(t)]) \leq 100d$.  
At each time $t$ during the phase, the algorithm will add and/or remove some vertices from the current vertex set~$W(t)$,
in order to reach the target. Next we describe the algorithm in pseudocode. 
In the algorithm, the sets $W^+(t)$ and $W^-(t)$ contain the vertices that still need to 
be added and removed, respectively. Initially, before the arrival of the first vertex,
we have $W^+(t) = W^-(t)=W(t)=\emptyset$.
\begin{algorithm}[H] 
\caption{\textsc{Set-and-Achieve-Target-IS}($v$)}
\begin{algorithmic}[1]
\State $\rhd$ $v$ is the vertex arriving at time $t$ and $G(t)$ is the graph after arrival of~$v$
\If{$W^+(t-1)=\emptyset$ and $W^-(t-1)=\emptyset$} \hfill $\rhd$ start a new phase
    \State Choose $V^*(t)\subseteq V$ such that $|V^*(t)|\geq \frac{99}{100}|V(t)|$ and $\Delta(G[V^*(t)]) \leq 100d$. 
    \State        Set $W^+(t) \gets V^*(t) \setminus W(t-1)$ and $W^-(t) \gets W(t-1) \setminus V^*(t)$.
\Else 
    \State Set $W^+(t) \gets W^+(t-1)$ and $W^-(t) \gets W^-(t-1)$
\EndIf
\State \label{step:IS-delete} Set $m^- \gets \min(1, |W^-(t)|)$. Delete a set $X$ of $m^-$ vertices from $W^-(t)$ and delete the same vertices from $W(t)$. 
       (Note that  $|X|\leq 1$. In the fully dynamic setting to be discussed later, however, we can also have $|X|=2$.)
\State Set $\Ialg(t) \gets \Ialg(t-1) \setminus X$  \label{step:delete-from-IS}
\State \label{step:IS-add} Set $m^+ \gets \min (1-m^-,|W^+(t)|)$. Delete $m^+$ vertices from $W^+(t)$ and add them to $W(t)$. 
\State \textsc{Greedy-Addition}($G(t),\Ialg(t), W(t)$) \label{step:add-to-IS}
\end{algorithmic}
\end{algorithm}
In the proofs below, $W^-(t)$ and $W^+(t)$ refer to the situation before Step~7.
\begin{lemma}\label{cardinality of W(t) at set phase}
If a new phase starts at time $t$, then $|W(t-1)| \geq \frac{495}{1000}\cdot |V(t-1)|$.
\end{lemma}
\begin{proof}
     We proceed by induction on~$t$. The lemma trivially holds at the start of the first
     phase, when $t=1$. Now consider the start of some later phase, at time~$t$,
     and let $\tprev$ be the previous time at which a new phase started. 
     Since the start of a phase is just after the end of the previous phase, 
     and at the end of a phase we have reached the target set, we have
     $|W(t-1)| \geq  \frac{99}{100}|V(\tprev)|$. 
     
     Since $W^+(\tprev)$ and $W^-(\tprev)$ are disjoint subsets of $V(\tprev)$, we have
     $|W^+(\tprev) \cup W^-(\tprev)| \leq |V(\tprev)|$. Since the vertices from 
     $W^+(\tprev)$ and $W^-(\tprev)$ are added resp.~deleted one at a time, this implies that
     $(t-\tprev) \leq |V(\tprev)|$. 
     So $|V(t-1)|\leq 2|V(\tprev)|$ and we have
     \[
   |W(t-1)| \geq  \frac{99}{100}\cdot|V(\tprev)|
          \geq  \frac{495}{1000}\cdot 2 \cdot |V(\tprev)| 
          \geq  \frac{495}{1000}\cdot |V(t-1)|      
     \]
This finishes the proof of the lemma.
\end{proof}

The previous lemma gives lower bound of $|W(t)|$ at the start of each phase. Next we use this
to give a lower bound of $|W(t)|$ at any time point $t$.
\begin{lemma} \label{lem:during-phase ind set}
For any time $t$, we have $|W(t)| \geq \frac{455}{1000}\cdot |V(t)|$.
\end{lemma}
\begin{proof}
Consider a time~$t$. Let $\tprev$ be the 
previous time at which a new phase started, and let $\tnext$ be the next time
at which a new phase starts.
    
Let $t^* := \max\{t': \tprev \leq t'< \tnext \mbox{ and } W^-(t')\neq\emptyset \}$.
In other words, $t^*$ is the last time step in the interval $[\tprev,\tnext)$ 
at which we still delete vertices from $W(t)$. If there are no vertices to
delete, that is, $W^-(\tprev)=\emptyset$, then we set $t^* := \tprev -1$.

Note that in the interval $[\tprev-1,\tnext)$, 
the value $\frac{|W(t)|}{|V(t)|}$ is minimum at $t=t^*$. This is true because
after time $\tprev-1$ until time $t*$ we are deleting vertices from $W(t)$ 
and adding vertices to $V(t)$, while after time $t^*$ we are adding exactly
one vertex to both $W(t)$ and $V(t)$ at each iteration.
and 
Since $|V^*(\tprev)|\geq \frac{99}{100}|V(\tprev)|$, we have
\[
|W^-(\tprev)| \leq |V(\tprev)\setminus V^*(\tprev)|   \leq \frac{1}{100} \cdot |V(\tprev)|.
\]
The vertices from $W^-(t)$ are deleted in $t^*-\tprev+1$ iterations, and so
$(t^*-\tprev+1) \leq \frac{1}{100} \cdot |V(\tprev)|$. 
Thus,
\[
|V(t^*)|\leq |V(\tprev-1)| + \frac{1}{100}\cdot|V(\tprev)| \mbox{ and } |W(t^*)|\geq |W(\tprev-1)| - \frac{1}{100} \cdot |V(\tprev)|.
\]
Hence, by Lemma~\ref{cardinality of W(t) at set phase} we have
\[
\begin{array}{lll}
|W(t^*)| & \geq & \frac{495}{1000}\cdot|V(\tprev-1)|-\frac{1}{100} \cdot |V(\tprev)| \\[2mm]
     & = & \frac{455}{1000}\cdot|V(\tprev-1)|+\frac{4}{100}\cdot|V(\tprev-1)|-\frac{1}{100} \cdot |V(\tprev)| \\[2mm]
     & \geq  & \frac{455}{1000}\cdot|V(\tprev-1)|+\frac{2}{100}\cdot|V(\tprev)|-\frac{1}{100} \cdot |V(\tprev)| \\[2mm]
 & \geq  &  \frac{455}{1000}\cdot |V(t^*)|.    \\[2mm]

\end{array}
\]
Note that the third line in the derivation uses that $|V(\tprev-1)| \geq \frac{1}{2}\cdot |V(\tprev)|$,
which holds when $|V(\tprev-1)|>0$ since $|V(\tprev)|=|V(\tprev-1)|+1$.
We can assume this because $|W(t^*)| \geq \frac{455}{1000}\cdot |V(t^*)|$ trivially holds when $V(\tprev-1)=\emptyset$.
Indeed, when $V(\tprev-1)=\emptyset$ then $t^*=\tprev-1$, and so $|W(t^*)| = |V(t^*)|=0$.
This finishes the proof of the lemma.
\end{proof}
The next lemma shows that our algorithm maintains a large independent set on $G[W(t)]$.
\begin{lemma} \label{lem:size of output ind set}
For any time $t$ we have $ |W(t)| = O(d)\cdot |\Ialg(t)|$.
\end{lemma}
\begin{proof}
Observe that at any time point $t$ the maximum degree of $G[W(t)]$ is bounded by~$100d$,
and that the independent set $\Ialg(t)$ maintained by the algorithm is a subset of~$W(t)$.
Choose a constant $c$ such that $c d > 100d+1$; for example $c=102$ will do.
The following claim guarantees that
when $\Ialg(t)$  becomes too small, then \textsc{Greedy-Addition} can add a vertex to it.
\begin{claiminproof}
If $I \subseteq W(t)$ is
an independent set with $|W(t)| > cd \cdot |I|$, then there exists a 
vertex $v \in W(t) \setminus I$ such that $\{v\} \cup I$ is a independent set. 
\end{claiminproof}
\begin{proofinproof}
Suppose for a contradiction that the claim is false.
Then all vertices in $W(t) \setminus I$ have an edge to a vertex in~$I$. 
Since $|W(t) \setminus I| \geq \frac{c d -1}{c d} \cdot |W(t)|$, 
this implies that the average degree of the vertices from~$I$ in 
$G[W(t)]$ is at least $c d-1>100d$, which contradicts that $\Delta(G[W(t)]) \leq 100d$.
\end{proofinproof}
We proceed to show by induction on~$t$ that $|W(t)| \leq cd\cdot |\Ialg(t)|$,
which proves the lemma.
The statement holds trivially for $t=0$. Now suppose it holds for $t=k$. 
There are two cases. 
\begin{itemize}
\item \emph{Case 1: At time $t=k+1$, a vertex $v^*$ is deleted from $W(t)$.} \\[1mm]
    In this case, after Step~\ref{step:delete-from-IS} we have $\Ialg(t)=\Ialg(t-1) \setminus X$,
    where $X=\{v^*\}$. 
    Now \textsc{Greedy-Addition} will try to add a new vertex to $\Ialg(t)$. 
    If $|W(t)| \leq cd\cdot |\Ialg(t)|$ already holds before \textsc{Greedy-Addition} is 
    executed, then we are done. Otherwise, the claim above guarantees that \textsc{Greedy-Addition} 
    will add a vertex to $\Ialg$, thus making up for the removal of~$v^*$. 
    Hence, $|\Ialg(t)| \geq |\Ialg(t-1)|$. Since $|W(t)|=|W(t-1)|-1$ and using the induction hypothesis, 
    we thus have 
    \[
    |W(t)| = |W(t-1)| -1 \leq cd\cdot |\Ialg(t-1)| - 1 \leq cd\cdot |\Ialg(t)| - 1 < cd\cdot |\Ialg(t)|.
    \]
\item \emph{Case 2: At time $t=k+1$, a vertex $v^*$ is added to $W(t)$.} \\[1mm]
    If $|W(t)| \leq cd\cdot |\Ialg(t)|$ already holds before \textsc{Greedy-Addition} is 
    executed, then we are done. Otherwise, our claim guarantees that \textsc{Greedy-Addition} 
    will add a vertex to~$\Ialg$. Thus $|\Ialg(t)|=|\Ialg(t-1)|+1$, 
    and so by induction we have
    \[
    |W(t)| = |W(t-1)| + 1 \leq cd \cdot |\Ialg(t-1)| + 1 = cd \cdot \left( |\Ialg(t)| - 1 \right) + 1\leq cd \cdot |\Ialg(t)|.
    \]
\end{itemize}
Hence, in both cases we have $|W(t)| \leq cd\cdot |\Ialg(t)|$.
\end{proof}
Putting everything together, we obtain the following result.
\begin{theorem}
There is a $2$-stable $O(d)$-approximation algorithm for \mis in the vertex-arrival model 
where the average degree of~$G(t)$ is bounded by $d$ at all times.
\end{theorem}
\begin{proof}
We know that $\opt(t)$, the size of a maximum independent set of $G(t)$,
is trivially bounded by $|V(t)|$. 
By Lemmas~\ref{lem:during-phase ind set} and~\ref{lem:size of output ind set} we thus have
\[
\opt(t) \geq \frac{1000}{455} \cdot |W(t)| = O(d) \cdot |\Ialg(t)|,
\]
Hence, the algorithm \textsc{set-and-achieve-target-IS} is an $O(d)$ approximation. 
The algorithm is clearly 2-stable, since it deletes at most one vertex from $\Ialg$ (in Step~\ref{step:delete-from-IS})
and it adds at most one vertex to $\Ialg$ (in Step~\ref{step:add-to-IS}).
\end{proof}

%% file: Deletions_Independent_Set.tex
\subparagraph*{A 6-stable $O(d)$-approximation algorithm in the fully dynamic setting.}
We now consider the fully dynamic model, where vertices can also be deleted.
More precisely, in each update either a vertex is added along with its incident edges, 
or a vertex is deleted along with its incident edges. Our algorithm for this setting
is very similar to the 2-stable algorithm for the insertion-only setting. The differences
are as follows.
\begin{itemize}
    \item First of all, whenever the update is the deletion of a vertex~$v$ from the graph,
          then $v$ is also deleted from any of the sets $W(t)$, $W^+(t)$, $W^-(t)$, and $\Ialg(t)$
          in which it is present.
    \item Second, in Steps~\ref{step:IS-delete} and~\ref{step:IS-add} of algorithm \textsc{Set-and-Achieve-Target-IS}, 
          we set  $m^- \gets \min(2, |W^-(t)|)$ and $m^+ \gets \min (2-m^-,|W^+(t)|)$. 
          (Note that now we can also have $|X|=2$.)
          In other words, we try to add or delete vertices from $W(t)$ in groups of two,
          instead of adding/deleting a single vertex. 
    \item Finally, instead of trying to add a single vertex to the independent set $I^*$,
          Algorithm \textsc{Greedy-Addition} will try to add three vertices to the independent set $I^*$, if possible.
          If not, it will try to add two vertices. If that fails as well, it will try to add only one vertex, if it also fails then no vertices are added.
\end{itemize}
Note that the stability has increased from $2$ to $6$ in order to handle deletions.
Indeed, the vertex being deleted from the graph may need to be removed from $\Ialg$,
we may remove two vertices from $\Ialg$ in Step~\ref{step:delete-from-IS},
and we may add three vertices to $\Ialg$ in Step~\ref{step:add-to-IS}.
\medskip

The analysis of the approximation ratio is very similar to the analysis in the insertion-only
setting, but slightly more subtle.  
We start with the equivalent of Lemma~\ref{cardinality of W(t) at set phase}.
\begin{lemma}\label{cardinality of W(t) at set phase with deletions}
If a new phase starts at time $t$, then $|W(t-1)| \geq \frac{98}{300}\cdot |V(t-1)|$.
\end{lemma}
\begin{proof}
     We proceed by induction on~$t$. The lemma trivially holds at the start of the first
     phase, when $t=1$. Now consider the start of some later phase, at time~$t$,
     and let $\tprev$ be the previous time at which a new phase started. Let $x$ denote the number of vertices 
     deleted from the graph~$G$
     during the time interval $(\tprev,t)$. Clearly $x \leq t-\tprev-1$.
     Since the start of a phase is just after the end of the previous phase,
     and at time $\tprev$ we have set ourselves a target of size $\frac{99}{100}|V(\tprev)|$,
     we have 
     \[
     |W(t-1)| \geq  \frac{99}{100}|V(\tprev)|-x. 
     \]     
     As before, since $W^+(\tprev)$ and $W^-(\tprev)$ are disjoint subsets of $V(\tprev)$, 
     we have $|W^+(\tprev) \cup W^-(\tprev)| \leq |V(\tprev)|$. This time, however, we make two changes
     at each update, so  
     \[
     (t-\tprev) \leq \ceil{|V(\tprev)|/2} \leq |V(\tprev)|/2+1.
     \]
     Note that $V(t-1)$ contains all vertices from $V(\tprev)$, plus the vertices that arrived from
     time $\tprev$ to time $t-1$, minus the vertices that have been deleted in this period.
     The number of deleted vertices is~$x$, and so the number of arriving vertices is~$t-\tprev-1-x$.
     Hence,
     \[
     |V(t-1)|\leq |V(\tprev)|+ (t-\tprev-1-x)-x \leq \frac{3}{2}\cdot |V(\tprev)|.
     \]
     Also, since $x \leq t-\tprev-1 \leq |V(\tprev)|/2$ 
     we have,
     \[
   |W(t-1)| \geq  \frac{99}{100}\cdot|V(\tprev)|-x
            \geq   \frac{98}{300}\cdot \frac{3}{2} \cdot |V(\tprev)| 
            \geq   \frac{98}{300}\cdot |V(t-1)|.    
     \]
This finishes the proof of the lemma.
\end{proof}
We use Lemma~\ref{cardinality of W(t) at set phase with deletions} 
to give a lower bound of $|W(t)|$ at any time point $t$.
\begin{lemma} \label{lem:during-phase ind set with deletions}
For any time $t$, we have $|W(t)| \geq \frac{89}{603}\cdot |V(t)|$.
\end{lemma}
\begin{proof}
Consider a time~$t$. Let $\tprev$ be the 
previous time at which a new phase started, and let $\tnext$ be the next time
at which a new phase starts.
    
Let $t^* := \max\{t': \tprev \leq t'< \tnext \mbox{ and } W^-(t')\neq\emptyset \}$, and
let $x$ denote the number of vertices deleted 
from the graph~$G$ during the time interval $(\tprev,t^*)$.
(As before, if there are no vertices to
delete, that is, $W^-(\tprev)=\emptyset$, then we set $t^* := \tprev -1$.) 

Like in the insertion-only setting,
the ratio $|W(t)|/|V(t)|$ attains its minimum during the interval $[\tprev,\tnext]$ at $t=t^*$.
Indeed, after time~$t^*$ we only delete a vertex from $W(t)$ when the vertex is deleted from $V(t)$,
which actually increases the ratio.
As before, we also have $|W^-(\tprev)| \leq \frac{1}{100} \cdot |V(\tprev)|$. 
Since we delete the vertices from $W^-(\tprev)$ in pairs, we have
\[
t^*-\tprev+1 \leq \left\lceil \frac{|V(\tprev)|/100}{2} \right\rceil \leq \frac{1}{200} \cdot |V(\tprev)|+1,
\] 
and so $|V(t^*)|\leq \frac{201}{200}|V(\tprev)|$.
By Lemma~\ref{cardinality of W(t) at set phase with deletions} we thus have,
\[
     \begin{array}{llll}
   |W(t^*)| & \geq & |W(\tprev-1)| - |W^-(\tprev)| - x  \\[2mm]
          &\geq & \frac{98}{300}\cdot|V(\tprev-1)|-\frac{1}{100} \cdot |V(\tprev)|-x  \\[2mm]
          &\geq& (\frac{49}{300}-\frac{1}{100})\cdot |V(\tprev)|-x \mbox{\ \ \ (since we can assume $|V(\tprev-1)| \geq \frac{1}{2}\cdot |V(\tprev)|$)}\\[2mm]
          &\geq& (\frac{46}{300}-\frac{1}{200})\cdot |V(\tprev)| \hfill \mbox{ (since $x\leq t^*-\tprev$)} \\[2mm]
          &\geq & \frac{89}{603}\cdot |V(t^*)|      
     \end{array}
\]

Note that the argument for the assumption $|V(\tprev-1)| \geq \frac{1}{2}\cdot |V(\tprev)|$ in the third step of the derivation
is the same as in the proof of Lemma~\ref{lem:during-phase ind set}.
This finishes the proof of the lemma.
\end{proof}
The next lemma is the equivalent of Lemma~\ref{lem:size of output ind set}.
Together with the previous lemma, it
implies that our algorithm gives an $O(d)$-approximation.

\begin{lemma} \label{lem:size of output ind set with deletions}
For any time $t$ we have $ |W(t)| = O(d)\cdot |\Ialg(t)|$.
\end{lemma}
\begin{proof}
As before, choose a constant $c$ such that $c \cdot d > 100d+1$
and note that the algorithm guarantees that the maximum degree of $G[W(t)]$ is bounded by~$100d$. 
The following claim still holds:
\begin{claiminproof}
If $I \subseteq W(t)$ is
an independent set with $|W(t)| > cd \cdot |I|$, then there exists a 
vertex $v \in W(t) \setminus I$ such that $\{v\} \cup I$ is a independent set. 
\end{claiminproof}
If $|W(t)| \leq cd \cdot |\Ialg(t)|$ before \textsc{Greedy-Addition} is called or after \textsc{Greedy-Addition} 
has added at most two vertices to $\Ialg(t)$, then we are done. Otherwise the claim above implies
that \textsc{Greedy-Addition} will add three vertices to $\Ialg(t)$.
We again proceed to show by induction on~$t$ that $|W(t)| \leq cd\cdot |\Ialg(t)|$,
which proves the lemma. The statement holds for $t=0$, so now suppose it holds for $t=k$. 

So at time $t=k+1$, let $x_i$ and $x_w$ denote the number of vertices deleted by adversary from $\Ialg(t-1)$ and $W(t-1)$ respectively. Also let $s^+$ denote the number of vertices added to $W(t-1)$ by $\alg$ and $s^-$ denote the number of vertices deleted from $W(t-1)$ by $\alg$.

Now at time $t=k+1$, after the execution of \textsc{Greedy-Addition}, by arguments above we know that either we have our lemma,
or 
\[
|\Ialg(t)|=|\Ialg(t-1)|+3-x_i-s^- \ \ \mbox{ and } \ \ |W(t)|=|W(t-1)|-x_w+s^+-s^-.
\]
Now since $\Ialg(t-1)$ is a subset of $W(t-1)$, we have $0 \leq x_i \leq x_w \leq 1$. We also know that $0 \leq s^+, s^- \leq 2$, and so 
 $3-x_i-s^-> -x_w+s^+-s^-$. Since $|W(t-1)|>|\Ialg(t-1)|$  we therefore have $\frac{|W(t)|}{|\Ialg(t)|} < \frac{|W(t-1)|}{|\Ialg(t-1)|} \leq c\cdot d$ (by induction).
This finishes the proof.
\end{proof}
Lemmas~\ref{lem:during-phase ind set with deletions} and~\ref{lem:size of output ind set with deletions}
imply the following theorem. 
\begin{theorem}
There is a $6$-stable $O(d)$-approximation algorithm for \mis in the fully-dynamic model
in the setting where the average degree of~$G(t)$ is bounded by $d$ at all times.
\end{theorem}

%% file: main.bbl
\begin{thebibliography}{10}

\bibitem{alon2023}
N.~Alon.
\newblock private communication, June 2023.

\bibitem{DBLP:conf/mfcs/0001DJ18}
Spyros Angelopoulos, Christoph D{\"{u}}rr, and Shendan Jin.
\newblock Online maximum matching with recourse.
\newblock In {\em Proc.~43rd International Symposium on Mathematical
  Foundations of Computer Science ({MFCS})}, volume 117 of {\em LIPIcs}, pages
  8:1--8:15, 2018.

\bibitem{DBLP:conf/esa/AntunesMM17}
Daniel Antunes, Claire Mathieu, and Nabil~H. Mustafa.
\newblock Combinatorics of local search: An optimal 4-local {H}all's theorem
  for planar graphs.
\newblock In {\em Proc.~25th Annual European Symposium on Algorithms ({ESA})},
  volume~87 of {\em LIPIcs}, pages 8:1--8:13, 2017.

\bibitem{berndt2019robust}
Sebastian Berndt, Valentin Dreismann, Kilian Grage, Klaus Jansen, and Ingmar
  Knof.
\newblock Robust online algorithms for certain dynamic packing problems.
\newblock In {\em Proc.~17th International Workshop on Approximation and Online
  Algorithms ({WAOA})}, volume 11926 of {\em Lecture Notes in Computer
  Science}, pages 43--59, 2019.

\bibitem{DBLP:conf/esa/BerndtEJLMR19}
Sebastian Berndt, Leah Epstein, Klaus Jansen, Asaf Levin, Marten Maack, and
  Lars Rohwedder.
\newblock Online bin covering with limited migration.
\newblock In {\em Proc.~27th Annual European Symposium on Algorithms ({ESA})},
  volume 144 of {\em LIPIcs}, pages 18:1--18:14, 2019.

\bibitem{DBLP:books/daglib/0097013}
Allan Borodin and Ran El{-}Yaniv.
\newblock {\em Online Computation and Competitive Analysis}.
\newblock Cambridge University Press, 1998.

\bibitem{DBLP:journals/algorithmica/BoyarEFKL19}
Joan Boyar, Stephan~J. Eidenbenz, Lene~M. Favrholdt, Michal Kotrbc{\'{\i}}k,
  and Kim~S. Larsen.
\newblock Online dominating set.
\newblock {\em Algorithmica}, 81(5):1938--1964, 2019.

\bibitem{DBLP:journals/algorithmica/BoyarFKL22}
Joan Boyar, Lene~M. Favrholdt, Michal Kotrbc{\'{\i}}k, and Kim~S. Larsen.
\newblock Relaxing the irrevocability requirement for online graph algorithms.
\newblock {\em Algorithmica}, 84(7):1916--1951, 2022.

\bibitem{DBLP:journals/iandc/ChlebikC08}
Miroslav Chleb{\'{\i}}k and Janka Chleb{\'{\i}}kov{\'{a}}.
\newblock Approximation hardness of dominating set problems in bounded degree
  graphs.
\newblock {\em Inf. Comput.}, 206(11):1264--1275, 2008.

\bibitem{DBLP:journals/corr/abs-2111-07812}
Minati De, Sambhav Khurana, and Satyam Singh.
\newblock Online dominating set and independent set.
\newblock {\em CoRR}, abs/2111.07812, 2021.

\bibitem{bbkmz-ethf-20}
Mark de~Berg, Hans~L. Bodlaender, S{\'{a}}ndor Kisfaludi{-}Bak, D{\'{a}}niel
  Marx, and Tom~C. van~der Zanden.
\newblock A framework for {Exponential-Time-Hypothesis}-tight algorithms and
  lower bounds in geometric intersection graphs.
\newblock {\em SIAM J. Comput.}, 49:1291--1331, 2020.

\bibitem{DBLP:conf/isaac/BergKMT21}
Mark de~Berg, S{\'{a}}ndor Kisfaludi{-}Bak, Morteza Monemizadeh, and Leonidas
  Theocharous.
\newblock Clique-based separators for geometric intersection graphs.
\newblock In {\em Proc.~32nd International Symposium on Algorithms and
  Computation (ISAAC)}, volume 212 of {\em LIPIcs}, pages 22:1--22:15, 2021.

\bibitem{DBLP:conf/swat/BergSS22}
Mark de~Berg, Arpan Sadhukhan, and Frits C.~R. Spieksma.
\newblock Stable approximation algorithms for the dynamic broadcast
  range-assignment problem.
\newblock {\em To appear in SIAM journal of Discrete Mathematics}.

\bibitem{DBLP:conf/icalp/0002HKSV14}
Oliver G{\"{o}}bel, Martin Hoefer, Thomas Kesselheim, Thomas Schleiden, and
  Berthold V{\"{o}}cking.
\newblock Online independent set beyond the worst-case: Secretaries, prophets,
  and periods.
\newblock In {\em Proc.~41st International Colloquium on Automata, Languages,
  and Programming ({ICALP})}, volume 8573 of {\em Lecture Notes in Computer
  Science}, pages 508--519, 2014.

\bibitem{DBLP:conf/stoc/GuptaK0P17}
Anupam Gupta, Ravishankar Krishnaswamy, Amit Kumar, and Debmalya Panigrahi.
\newblock Online and dynamic algorithms for set cover.
\newblock In {\em Proc.~49th Annual {ACM} {SIGACT} Symposium on Theory of
  Computing ({STOC})}, pages 537--550, 2017.

\bibitem{DBLP:conf/soda/GuptaKS14}
Anupam Gupta, Amit Kumar, and Cliff Stein.
\newblock Maintaining assignments online: Matching, scheduling, and flows.
\newblock In Chandra Chekuri, editor, {\em Proc.~2th Annual {ACM-SIAM}
  Symposium on Discrete Algorithms ({SODA})}, pages 468--479, 2014.

\bibitem{DBLP:journals/tcs/HalldorssonIMT02}
Magn{\'{u}}s~M. Halld{\'{o}}rsson, Kazuo Iwama, Shuichi Miyazaki, and Shiro
  Taketomi.
\newblock Online independent sets.
\newblock {\em Theor. Comput. Sci.}, 289(2):953--962, 2002.

\bibitem{DBLP:conf/esa/Har-PeledQ15}
Sariel Har{-}Peled and Kent Quanrud.
\newblock Approximation algorithms for polynomial-expansion and low-density
  graphs.
\newblock In {\em Proc.~23rd Annual European Symposium on Algorithms ({ESA})},
  volume 9294 of {\em Lecture Notes in Computer Science}, pages 717--728.
  Springer, 2015.

\bibitem{DBLP:journals/ipl/KingT97}
Gow{-}Hsing King and Wen{-}Guey Tzeng.
\newblock On-line algorithms for the dominating set problem.
\newblock {\em Inf. Process. Lett.}, 61(1):11--14, 1997.

\bibitem{LT-planar-separator-thm}
Richard~J. Lipton and Robert~Endre Tarjan.
\newblock A separator theorem for planar graphs.
\newblock {\em {SIAM} J. Appl. Math}, 36(2):177--189, 1977.

\bibitem{DBLP:conf/soda/LiptonT94}
Richard~J. Lipton and Andrew Tomkins.
\newblock Online interval scheduling.
\newblock In {\em Proc.~5th Annual {ACM-SIAM} Symposium on Discrete Algorithms
  ({SODA)}}, pages 302--311, 1994.

\bibitem{10.1007/978-3-031-18367-6_7}
Alison Hsiang-Hsuan Liu and Jonathan Toole-Charignon.
\newblock The power of amortized recourse for online graph problems.
\newblock In Parinya Chalermsook and Bundit Laekhanukit, editors, {\em
  Approximation and Online Algorithms}, pages 134--153, Cham, 2022. Springer
  International Publishing.

\bibitem{DBLP:conf/stoc/Zuckerman06}
David Zuckerman.
\newblock Linear degree extractors and the inapproximability of max clique and
  chromatic number.
\newblock In {\em Proc.~38th Annual {ACM} Symposium on Theory of Computing
  ({STOC)}}, pages 681--690, 2006.

\end{thebibliography}
